\documentclass[12pt]{article}
\pdfoutput=1

\usepackage[margin=2.8cm]{geometry} 
\usepackage{fleqn} 

\usepackage{hyperref}
\usepackage{graphics}
\usepackage{graphicx}
\usepackage{subfigure}
\usepackage{float}
\usepackage{epstopdf}
\usepackage{epsfig}
\usepackage{color}
 \expandafter\let\csname equation*\endcsname\relax
  \expandafter\let\csname endequation*\endcsname\relax
 \usepackage{amsmath}
\usepackage{amsthm}
\usepackage{amssymb,amsfonts}
\usepackage{booktabs}
\usepackage[T1]{fontenc}
\usepackage{etoolbox}

\usepackage[utf8]{inputenc}
\usepackage[english]{babel}
\usepackage{tikz}

\newtheorem{theorem}{Theorem}[section]
\newtheorem{lemma}[theorem]{Lemma}
\newtheorem{corollary}[theorem]{Corollary}
\newtheorem{remark}[theorem]{Remark}

\newcommand{\ppf}[2]{\frac{\partial #1}{\partial #2}}

\title{The role of the time delay in the reflection and transmission of ultrashort 
electromagnetic pulses on a system of parallel current sheets} 

\author{
  M\'onika~Polner\thanks{Bolyai Institute, University of Szeged, H-6720 Szeged, Aradi v\'ertan\'uk tere 1, Hungary and ELI-ALPS, ELI-HU Ltd, Dugonics t\' er 13, Szeged 6720, Hungary.
    (\href{mailto:polner@math.u-szeged.hu}{polner@math.u-szeged.hu})}
  \, and 
  S\'andor Varr\'o\thanks{Institute for Solid State Physics and Optics, Wigner Research Center for Physics of the Hungarian Academy of Sciences, Budapest, Hungary and ELI-ALPS, ELI-HU Ltd, Dugonics t\' er 13, Szeged 6720, Hungary.}
  \, and
Anett V\"or\"os-Kiss\thanks{Bolyai Institute, University of Szeged, H-6720 Szeged, Aradi v\'ertan\'uk tere 1, Hungary.}
}

\begin{document}



\maketitle 






\begin{abstract} 
The reflection and transmission of a few-cycle laser pulse impinging on two parallel thin metal layers have been analyzed. 
The two layers, with a thickness much smaller than the skin depth of the incoming radiation field, are represented by current sheets embedded in three dielectrics, all with different index of refraction. 
The dynamics of the surface currents and the scattered radiation field are described by the coupled system of Maxwell--Lorentz equations. When applying the plane wave modeling assumptions, these reduce to a 
hybrid system of two delay differential equations for the electron motion in the layers and a 
recurrence relation for the scattered field. The solution is given as the limit of a singularly perturbed system and the effects of the time delay
on the dynamics is analyzed.
\end{abstract}

\vspace{2pc}
\noindent{\it Keywords}: Scattering, Maxwell--Lorentz equations, radiation reaction, surface current, 
delay differential equations, difference equations, singularly perturbed system 

\section{Introduction}
The scattering of ultrashort electromagnetic pulses in a system of two (or more) parallel current sheets is physically significant and the solution of the governing system of equations is also a nontrivial mathematical challenge. 
This paper gives a theoretical description of the reflection and transmission of a few-cycle laser pulse impinging on two thin metal layers, represented by surface currents. The mathematical analysis of this problem in the time-domain is based on the theory of delay differential equations \cite{Hale,DiekmannRFDE}. The first description of such a system was given by Sommerfeld \cite{Sommerfeld}, where the temporal distortion of x-ray pulses impinging perpendicularly on one surface in vacuum was analyzed.
This was then subsequently generalized in \cite{Varro2004} by allowing oblique incidence of the incoming radiation field and embedding the surface current in two semi-infinite dielectrics with two different indices of refraction. 
This general system 
 has been investigated from several physical points of views \cite{Varro2007} and the relativistic dynamics of the surface current has also been discussed  \cite{Varro2007carrier}.

The model described in this paper is an extension of the one-layer scattering problem applied to more layers, with an analysis based on classical electrodynamics and non-relativistic mechanics.       

Two parallel metal layers, with thickness much smaller than the skin depth of the radiation field are considered and represented by current sheets, embedded in three dielectrics, all with different index of refraction, 
see Figure~\ref{phen}. The target defined this  way can be imagined as a thin metal layer evaporated, for instance, 
on a glass substrate. The reflection and transmission of a few-cycle laser pulse impinging on the system of the two thin metal layers is studied. The dynamics of the surface currents and the complete radiation field are described by a coupled system of Maxwell equations and the equations of motion of the electrons which move in two parallel planes.
 The planar symmetry of the system (translation invariance along the interfaces) means that the spatial dependence of the Maxwell fields can be considerably simplified if the incoming and scattered fields are modeled by plane waves. 
 This corresponds to a one-dimensional propagation along the normals of the equal-phase planes and thus, the original partial differential equations reduce to ordinary differential equations with respect to the common retarded time. 
In this way, a hybrid system (HS) of equations is obtained that combines a system of delay differential 
equations (DDE) for the electron velocities, with a difference equation for the reflected wave stemming from 
the second surface. 
Compared to the previous studies on the one-layer problem, 
placing an additional metal layer between dielectrics induces time delays in the system. The sizes of the time delays depend on the distance between the two surface current sheets, the indices of refraction 
of the dielectrics they are embedded in and the angle of incidence of the impinging plane wave. The main result of this paper is that we can solve the HS for an arbitrary intensity, that is admissible in the linear approximation, and shape of electromagnetic radiation pulse. The solution is obtained as the limit of the solution of a singularly perturbed system and this is a new result. In the classical description of such scattering problems, the resulting model system is Fourier transformed and further analyzed in the frequency domain. The aim of this paper is to describe the temporal 
behavior of the field strength of the reflected and transmitted signals.

The most remarkable feature of this model is that a collective radiation-reaction term is automatically 
derived in the closed system of equations for the surface current. Damping terms appear naturally in the model, as has been first derived in the one-layer problem \cite{Varro2004} and are governed, besides the elementary charge and the electron mass, by the electron density. The appearance of these damping terms is a result of the back-action of the radiation field on the assembly of electrons, which derives from the boundary conditions, so they differ from friction-like forces.

The outline of this paper is as follows. Section~\ref{s:physical_model} presents the basic equations describing the model and in Section~\ref{s:solutionDDS}, the solution of the resulting HS is given as a limit of the solution of a singularly perturbed system.
 Time simulations are performed in Section~\ref{s:simulations} to illustrate the long time behavior of the solution, and a short discussion on the spectrum of the reflected wave is given. This illustrates how the intensity depends on the carrier-envelope (CE) phase difference of the incoming few-cycle pulse. 
\section{The basic equations of the model}\label{s:physical_model}
The model derived in this section, is an extension of the one layer problem studied in \cite{Varro2004}, in the sense that the main construction steps of the mathematical model are the same. This
 results in a hybrid system which requires a careful analysis of its qualitative properties.

Consider the following geometrical setup in the $(x,y,z)$ coordinate system. 
The first dielectric, with index of refraction $n_1$,  fills the region $z>l_2/2$ in space -- region~1. 
In region 2 a thin metal layer of thickness $l_2$,  is placed perpendicular to the $z$-axis around the $z = 0$ position, occupying the space region defined by the relation $-l_2/2\leq z\leq l_2/2$. Region 3, $-h+l_4/2<z<-l_2/2$,  is assumed to be filled by the second  dielectric, with index of refraction $n_3$. 
The second metal layer, with thickness $l_4$,  is placed  perpendicular to the $z$-axis around $z=-h$,   and this fills region~4, defined by $-h-l_4/2<z<-h+l_4/2$.  
Finally, region 5  is the dielectric with index of refraction $n_5$ occupying the region $z<-h-l_4/2$. The plane of incidence is defined as  the $yz$-plane and the initial $k$-vector is assumed to make an angle $\theta_1$ with the $z$-axis. This is shown in Figure~\ref{phen}.
\begin{figure}
\centering
\includegraphics[width=.4\columnwidth]{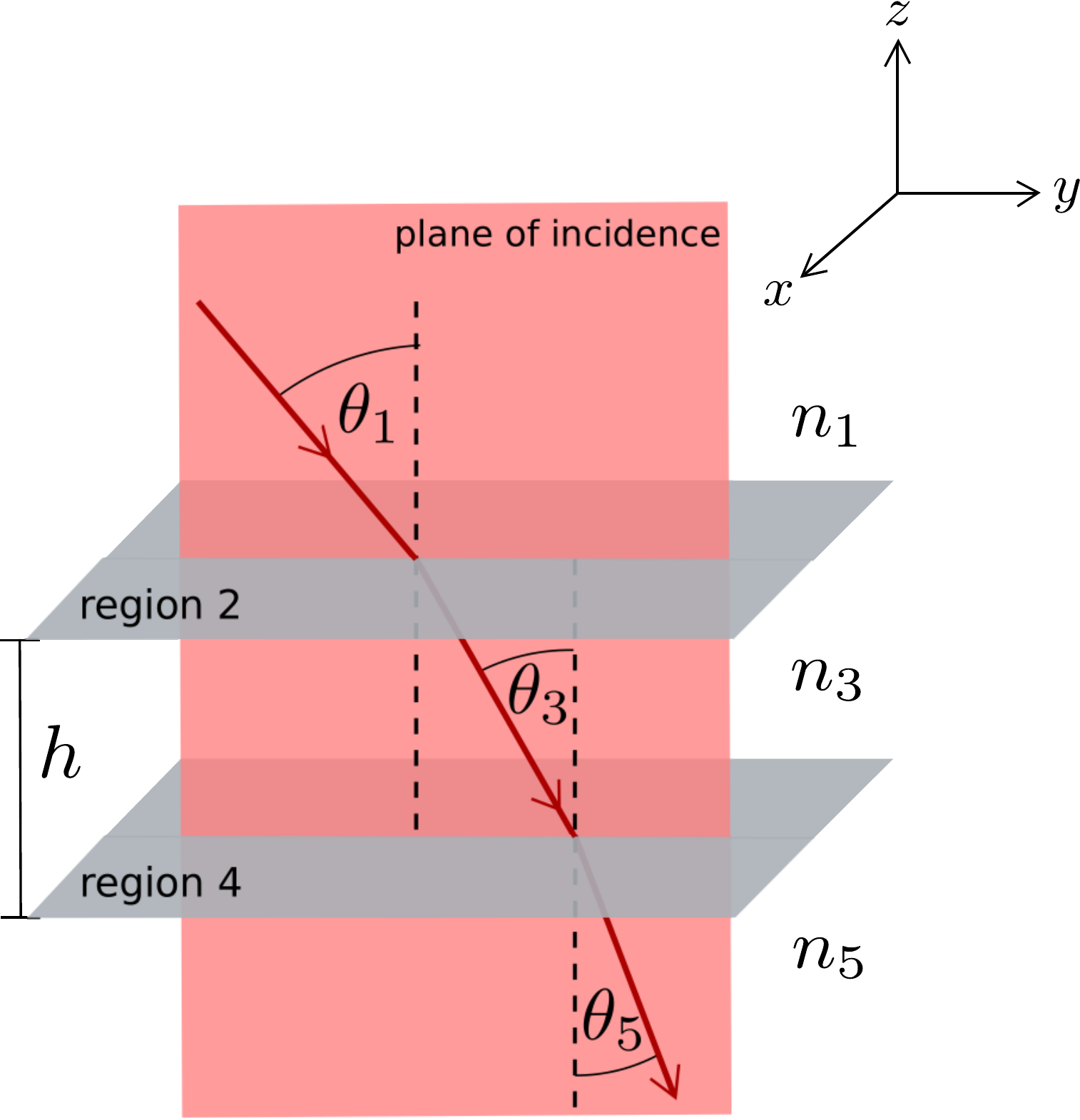}
\caption[]{The schematic view of the physical problem.}\label{phen} 
\end{figure}

In regions 1, 3 and 5, the field equations for a TM (p-polarized) wave, i.e., with the electric field and 
magnetic induction components $E=(0,E_y,E_z)$ and $B=(B_x,0,0)$, respectively, satisfy the Maxwell equations. This is, in cgs units, 
\begin{equation}
\partial_y{E_z}-\partial_z{E_y}=-\partial_0{B_x},\ 
\partial_z{B_x}=n^2\partial_0{E_y}+\mu\frac{4\pi}{c}J_y,\ 
-\partial_y{B_x}=n^2\partial_0{E_z}+\mu\frac{4\pi}{c}J_z\label{Bxyfield}
\end{equation}
and $\partial_x{E_z}=\partial_x{E_y}=0$. Here, $n=\sqrt{\mu \epsilon}$ is the index of refraction, 
with $\epsilon$ and $\mu$ the dielectric constant (with no dimension) and magnetic permeability (with no dimension), respectively, and $J$ is the electric current density. 
The following notations are used for the partial derivatives 
\[
\partial_0:=\frac{1}{c}\ppf{}{t},\quad \partial_x:=\ppf{}{x},\quad 
\partial_y:=\ppf{}{y},\quad \partial_z:=\ppf{}{z},
\]
with $c$ the speed of light in vacuum. 
The first step towards the model construction is to observe that in region 1 the $x$-component of the 
magnetic induction $B_{1x}$ satisfies the wave equation (no current, i.e., $J=0$)
\begin{equation}\label{wave_r1}
\left(\partial_y^2+\partial_z^2\right) B_{1x}=n_1^2\partial_0^2 B_{1x},
\end{equation}
where the subscript 1 refers to region 1 and $n_1$ is the refractive index here. The solution of 
\eqref{wave_r1} has the form
\begin{equation}
B_{1x}(r,t)=B_{1x}\Bigl(t-n_1\frac{r\cdot s}{c}\Bigr),
\end{equation}
where $r=(x,y,z)$ denotes the position vector and $s$ is the direction of wave propagation. In 
region 1$, B_{1x}$ is taken to be the superposition of the given incoming plane wave pulse $F$ and an 
unknown reflected plane wave $f_1$
\begin{equation}\label{wave_solution}
B_{1x}(y,z,t)=F\Bigl(t-n_1\frac{y\sin\theta_1-z\cos\theta_1}{c}\Bigr)-
f_1\Bigl(t-n_1\frac{y\sin\theta_1+z\cos\theta_1}{c}\Bigr).
\end{equation}
Here we used that $F$ propagates in the direction $(0,\sin\theta_1,-\cos\theta_1)$, with 
$\theta_1$ the angle of incidence, while the reflected $f_1$ wave propagates in the $(0,\sin\theta_1,\cos\theta_1)$ direction.  

The components $E_{1y},E_{1z}$ of the electric field in region 1 can be written, using Maxwell's equations,  in terms  of $F$ and $f_1$, and the partial derivative of $B_{1x}$ in \eqref{wave_solution} can be calculated using the chain rule to obtain 
\begin{equation}
\partial_z{B_{1x}}= \frac{n_1\cos\theta_1}{c}\ppf{}{t}(F+f_1)=n_1\cos\theta_1\partial_0(F+f_1).
\end{equation}
The second equation in \eqref{Bxyfield} is  now equivalent with 
\[
n_1\cos\theta_1\partial_0(F+f_1)=n_1^2\partial_0 E_{1y}.
\] 
Rearranging, yields $\cos\theta_1(F+f_1)- n_1E_{1y}=K$, where $K$ is a constant. The electric field components are obtained when $K=0$ as 
\begin{equation}\label{Eyfield_r1}
E_{1y}(y,z,t)=\frac{\cos\theta_1}{n_1}\left(F\Bigl(t-n_1\frac{y\sin\theta_1-z\cos\theta_1}{c}\Bigr)
+f_1\Bigl(t-n_1\frac{y\sin\theta_1+z\cos\theta_1}{c}\Bigr)\right),
\end{equation}
and similarly
\begin{equation*}
E_{1z}(y,z,t)=\frac{\sin\theta_1}{n_1}\left(F\Bigl(t-n_1\frac{y\sin\theta_1-z\cos\theta_1}{c}\Bigr)
-f_1\Bigl(t-n_1\frac{y\sin\theta_1+z\cos\theta_1}{c}\Bigr)\right).
\end{equation*} 
The magnetic induction $B_{3x}$  in region 3 can be written as 
the superposition of the unknown refracted wave $g_3$ and the reflected $f_3$ wave stemming from surface 4
\begin{equation}\label{wave_solution_region3}
B_{3x}(y,z,t)=g_3\Bigl(t-n_3\frac{y\sin\theta_3-z\cos\theta_3}{c}\Bigr)-
f_3\Bigl(t-n_3\frac{y\sin\theta_3+z\cos\theta_3}{c}\Bigr),
\end{equation}
with $\theta_3$ the refraction angle. Similar to region 1, from Maxwell's equations and \eqref{wave_solution_region3}, 
the components $E_{3y},E_{3z}$ of the electric field in region 3 can be expressed in terms of $f_3$ and $g_3$ as
\begin{align}
E_{3y}(y,z,t)&=\frac{\cos\theta_3}{n_3}\left(g_3\Bigl(t-n_3\frac{y\sin\theta_3-z\cos\theta_3}{c}\Bigr)
+f_3\Bigl(t-n_3\frac{y\sin\theta_3+z\cos\theta_3}{c}\Bigr)\right),\label{Eyfield_r3}\\
E_{3z}(y,z,t)&=\frac{\sin\theta_3}{n_3}\left(g_3\Bigl(t-n_3\frac{y\sin\theta_3-z\cos\theta_3}{c}\Bigr)
-f_3\Bigl(t-n_3\frac{y\sin\theta_3+z\cos\theta_3}{c}\Bigr)\right).
\end{align} 
Finally, in region 5, due to the absence of any reflecting surface, we express $B_{5x}$ and $E_{5y},E_{5z}$ in terms of the unknown refracted wave $g_5$
\begin{eqnarray}
B_{5x}(y,z,t)&=g_5\Bigl(t-n_5\frac{y\sin\theta_5-z\cos\theta_5}{c}\Bigr),\\
E_{5y}(y,z,t)&=\frac{\cos\theta_5}{n_5}g_5\Bigl(t-n_5\frac{y\sin\theta_5-z\cos\theta_5}{c}\Bigr),\label{Eyfield_r5}\\
E_{5z}(y,z,t)&=\frac{\sin\theta_5}{n_5}g_5\Bigl(t-n_5\frac{y\sin\theta_5-z\cos\theta_5}{c}\Bigr),
\end{eqnarray}
where $\theta_5$ is the refraction angle. 
Region 2 is the thin plain layer of thickness $l_2$. Maxwell's equations in this region yield
\begin{equation}
\partial_y{E_z}-\partial_z{E_y}=-\partial_0{B_x},\quad
\partial_z{B_x}=n_2^2\partial_0{E_y}+\frac{4\pi}{c}J_{2y}\label{Bxzfield_r2}.
\end{equation}
The boundary conditions for the field components are obtained by integrating both equations in \eqref{Bxzfield_r2} with respect to $z$ on the interval $[-l_2/2,l_2/2]$ and then 
taking the limit $l_2\to 0$, 
\begin{eqnarray}
&\left[E_{1y}-E_{3y}\right]\mid_{z=0}=0,\label{matchingE13}\\
&\left[B_{1x}-B_{3x}\right]\mid_{z=0}=\frac{4\pi}{c}\lim_{l_2\to 0}\int_{-l_2/2}^{l_2/2}J_{2y}dz=\frac{4\pi}{c}K_{2y},\label{Bfield_r2}
\end{eqnarray}
where $K_{2y}$ is the $y$-component of the surface current in layer 2. This means that the jump 
 in the electric field  components through the layers is zero and the jump in the magnetic  field components induces the surface current $K_{2y}$, which can further be expressed in terms of the local 
velocity of the electrons in the metal film
\begin{equation}
K_{2y}=e(\frac{d{\delta_y}_2}{dt})l_2 {n_e}_2.
\end{equation}
Here $e$ is the electron charge, ${n_{e}}_2$ is the density of electrons in the layer and 
${\delta_y}_2$ is the local displacement of the electrons in the $y$-direction. 
The right hand side of \eqref{Bfield_r2} can be written as
\begin{equation}
\frac{4\pi}{2c}K_{2y}=\frac{m}{e}\Gamma_2\frac{d{\delta_y}_2}{dt},
\end{equation}
where $m$ is the electron's mass and 
\begin{equation}
\Gamma_2=2\pi\frac{e^2}{m c}l_2{n_e}_2.
\end{equation}
Note that the parameter $\Gamma_2$ has dimension of frequency and its physical meaning will be 
a damping factor in the equation of motion of the electrons coupled with the radiation field. Let us write
\begin{equation} \label{damping2}
\Gamma_2=\left(\frac{\omega_{p_2}}{\omega_0}\right)^2\frac{\pi l_2}{\lambda_0}\omega_0,
\end{equation}
where $\omega_0$,  $\lambda_0=2\pi c/\omega_0$ are the carrier frequency and the central wavelength of the 
incoming light pulse, respectively and $\omega_{p_2}$ denotes the plasma frequency in the first metal layer.   

The electric field components are completely described in \eqref{Eyfield_r1} and \eqref{Eyfield_r3} 
for regions 1 and 3, respectively. Hence, the matching condition \eqref{matchingE13} is equivalent 
with
\begin{align}
E_{1y}(y,0,t)&=\frac{\cos\theta_1}{n_1}\left[F\Bigl(t-n_1\frac{y\sin\theta_1}{c}\Bigr)
+f_1\Bigl(t-n_1\frac{y\sin\theta_1}{c}\Bigr)\right]\nonumber\\
&=\frac{\cos\theta_3}{n_3}\left[g_3\Bigl(t-n_3\frac{y\sin\theta_3}{c}\Bigr)
+f_3\Bigl(t-n_3\frac{y\sin\theta_3}{c}\Bigr)\right]=E_{3y}(y,0,t).\label{matchingE13_2}
\end{align} 
This boundary, or matching condition enforces Snell's law of refraction to hold 
$n_1\sin\theta_1=n_3\sin\theta_3$.  Thus, when the common retarded time is introduced at the surface
\begin{equation}\label{retarded_time}
t'=t-\frac{n_iy\sin\theta_i}{c}, \ i=1,3,
\end{equation} 
\eqref{matchingE13_2} takes the form
\begin{equation}\label{f1g3f3}
c_1\left(F(t')+f_1(t')\right)=c_3\left(g_3(t')+f_3(t')\right),
\end{equation}
with $c_i=\cos\theta_i/n_i,\ i=1,3$. 

The magnetic field components are also described in \eqref{wave_solution} and 
\eqref{wave_solution_region3} for regions 1 and 3, respectively, hence the matching condition \eqref{Bfield_r2} is equivalent with 
\begin{align}
B_{1x}(y,0,t)-B_{3x}(y,0,t)&=\left[F\Bigl(t-n_1\frac{y\sin\theta_1}{c}\Bigr)
-f_1\Bigl(t-n_1\frac{y\sin\theta_1}{c}\Bigr)\right]\nonumber\\
&-\left[g_3\Bigl(t-n_3\frac{y\sin\theta_3}{c}\Bigr)
-f_3\Bigl(t-n_3\frac{y\sin\theta_3}{c}\Bigr)\right]=\frac{4\pi}{c}K_{2y}.\label{matchingB13}
\end{align} 
In terms of the retarded time $t'$,  \eqref{matchingB13} is
\begin{equation}\label{f1g3f3Ky}
F(t')-f_1(t')-\left(g_3(t')-f_3(t')\right)=\frac{4\pi}{c}K_{2y}(t').
\end{equation}
Using the same procedure in region 4 as in region 2, the following boundary conditions can be obtained
\begin{eqnarray}
&\left[E_{3y}-E_{5y}\right]\mid_{z=-h}=0\label{matchingE35}\\
&\left[B_{3x}-B_{5x}\right]\mid_{z=-h}=\frac{4\pi}{c}\int_{-h-l_4/2}^{-h+l_4/2}
J_{4y}dz
=\frac{4\pi}{c}K_{4y},\label{Bfield_r4}
\end{eqnarray}
where $K_{4y}$ is the $y$-component of the surface current in layer 4. 
From \eqref{Eyfield_r3} and \eqref{Eyfield_r5} it follows that \eqref{matchingE35} is 
equivalent with 
\begin{align}
E_{3z}(y,-h,t)&=\frac{\cos\theta_3}{n_3}\left[g_3\Bigl(t-n_3\frac{y\sin\theta_3+h\cos\theta_3}{c}\Bigr)
+f_3\Bigl(t-n_3\frac{y\sin\theta_3-h\cos\theta_3}{c}\Bigr)\right]\nonumber\\
&=\frac{\cos\theta_5}{n_5}g_5\Bigl(t-n_5\frac{y\sin\theta_5+h\cos\theta_5}{c}\Bigr)=E_{5y}(y,-h,t).\label{matchingE35_2}
\end{align}
This matching condition implies Snell's law to hold $n_3\sin\theta_3=n_5\sin\theta_5$,  
hence \eqref{matchingE35_2} is equivalent with 
\begin{equation}\label{g3f3g5}
c_3\left(g_3(t'-\Delta t_3)+f_3(t'+\Delta t_3)\right)=c_5g_5(t'-\Delta t_5),
\end{equation}
with $c_5=\cos\theta_5/n_5$ and
\begin{equation}\label{delays}
\Delta t_3=n_3 \frac{h\cos\theta_3}{c},\quad \Delta t_5=n_5 \frac{h\cos\theta_5}{c}.
\end{equation}
The time delay $\Delta t_i, i=3,5$ represents the time it takes for the signal to propagate a distance $h\cos\theta_i$ 
in the media with index of refraction $n_i.$ The delay times play an important role in our analysis.

Similarly, \eqref{Bfield_r4} is 
\begin{equation}\label{g3f3K4}
g_3(t'-\Delta t_3)-f_3(t'+\Delta t_3)-g_5(t'-\Delta t_5)=\frac{4\pi}{c}K_{4y}(t').
\end{equation}
Equations \eqref{f1g3f3}, \eqref{f1g3f3Ky}, \eqref{g3f3g5} and \eqref{g3f3K4} mean four linear relations for the six unknown
functions $f_1$, $f_3$, $g_3$, $g_5$, $K_{2y}$ and $K_{4y}$, so they are not enough to determine 
for instance the reflected wave $f_1$ and the transmitted wave $g_5$. The additional two relations are
given by the equation of motion for the surface currents, or more precisely for the
velocity components $\frac{d {\delta_y}_2}{dt}$ and $\frac{d {\delta_y}_4}{dt}$. In the non-relativistic regime, these equations are
\begin{eqnarray}
m\frac{d^2 {\delta_y}_2}{dt'\,^2}&=e E_{1y}\mid_{z=0}=ec_1\left[F(t')+f_1(t')\right],\\[5pt]
m\frac{d^2 {\delta_y}_4}{dt'\,^2}&=e E_{3y}\mid_{z=-h}=ec_3\left[g_3(t'-\Delta t_3)+f_3(t'+\Delta t_3)\right].
\end{eqnarray}

The resulting coupled system consists of a recurrence relation for $f_3$ and two delay differential 
equations for the local displacements $\delta_{y_2}$ and $\delta_{y_4}$ of the electrons in the metal layers:
\begin{subequations}\label{eq:282930}
\begin{align}
f_3(t')&= \frac{c_5-c_3}{c_5+c_3}
\cdot \frac{c_1-c_3}{c_1+c_3} f_3(t'-2\Delta t_3)
\label{eq:28}
\nonumber
\\[5pt]
&+ \frac{c_5-c_3}{c_5+c_3}
\cdot \frac{2c_1}{c_1+c_3}
\left[ F(t'-2\Delta t_3)
- \frac{m}{e} \Gamma_2 \dot\delta_{y_2}(t'-2\Delta t_3)
\right]
\nonumber
\\[5pt]
&- \frac{2c_5}{c_5+c_3} \cdot \frac{m}{e} \Gamma_4 \dot\delta_{y_4}(t'-\Delta t_3),
\\[5pt]
\ddot{\delta}_{y_2}(t')
&= \frac{2c_1 c_3}{c_1+c_3}
\left[
\frac{e}{m}F(t') - \Gamma_2 \dot{\delta}_{y_2}(t')
+ \frac{e}{m} f_3(t')
\right],
\label{eq:29}
\\[5pt]
\ddot{\delta}_{y_4}(t')
&= \frac{2c_1 c_3}{c_1+c_3}
\left[
\frac{e}{m}F(t'-\Delta t_3) - \Gamma_2 \dot\delta_{y_2}(t'-\Delta t_3)
\right]
\label{eq:30}
\nonumber
\\[5pt]
&+ \frac{c_1-c_3}{c_1+c_3} c_3 \frac{e}{m}
f_3(t'-\Delta t_3)
+ c_3 \frac{e}{m} f_3(t'+\Delta t_3),
\end{align}
\end{subequations}
where dots denote the derivatives with respect to the retarded time $t'$, see \eqref{retarded_time}. The two equations of motion (\ref{eq:29}) and (\ref{eq:30}) together with the recurrence relation (\ref{eq:28}) (delay difference equation) constitute a closed system of equations for the three unknown functions. Once these functions are known, the reflected wave $f_1$ and the transmitted wave $g_5$ can be calculated as 
\begin{equation}\label{eq:31_refl}
f_1(t')
= \frac{1}{c_1+c_3}
\left[(c_3-c_1)F(t')
-2 c_3
\frac{m}{e} \Gamma_2 \dot{\delta}_{y_2}(t')
+2 c_3 f_3(t')
\right],
\end{equation}
\begin{align}
g_5(t')
&= \frac{2 c_1}{c_1+c_3}
\left[F(t'+\Delta t_5 -\Delta t_3)
-
\frac{m}{e} \Gamma_2 \dot{\delta}_{y_2}(t'+\Delta t_5 -\Delta t_3)
\right]\nonumber\\[5pt]
&+\frac{c_1-c_3}{c_1+c_3} f_3(t'+\Delta t_5 -\Delta t_3)
-f_3(t'+\Delta t_5 +\Delta t_3)
-\frac{m}{e} \Gamma_4 \dot{\delta}_{y_4}(t'+\Delta t_5)
.
\label{eq:31_tran}
\end{align}     
It is quite remarkable that the damping terms, being proportional with $\Gamma_2, \Gamma_4$,   are automatically included 
in the system, without assuming any phenomenological friction. The appearance of the 
damping term is a manifestation of the radiation reaction coming from the boundary conditions. Since 
$\Gamma_2,\Gamma_4$ are proportional with the electron densities in the layers, this effect is due to the 
collective response of the electrons to the action of the complete radiation field, which on the other 
hand reacts back to the electrons.

It is possible to make the equations dimensionless by introducing the dimensionless vector potential $a_0$ 
defined in \ref{s:dimensionless}. This form of the system is the starting point of our mathematical analysis.
\section{The solution of the hybrid system}\label{s:solutionDDS}
In this section, the solution of the coupled HS \eqref{eq:282930} is given in 
its dimensionless form \eqref{eq:36_5}--\eqref{eq:f3_dimless}, using the theory of singularly perturbed systems, 
\cite{cooke65, cooke66}.

Consider \eqref{eq:36_5}--\eqref{eq:f3_dimless} in the form 
\begin{subequations}\label{DDEf3}
\begin{align}
\dot x_{1}(t)
&= a_1 \left[ -r_2x_1(t)+a_0 x_3(t)+a_0 F(t)\right],
\label{x1}\\[7pt]
\dot x_{2}(t)
&= a_5\left[-(a_1-a_2) r_4x_2(t)-a_1r_2x_1(t-\tau)+a_0 a_2 x_3(t-\tau)+ a_0 a_1F(t-\tau)\right],\label{x2}\\[7pt]
\tiny x_3(t)
&=\frac{a_2 (a_5-1)}{a_1-a_2}x_3(t-2\tau)
+\frac{a_1 (a_5-1)}{a_1-a_2}\left[F(t-2\tau)-\frac{r_2}{a_0}x_1(t-2\tau)\right]\nonumber\\
&\quad-a_5 \frac{r_4}{a_0}x_2(t-\tau),\label{x3}
\end{align}
\end{subequations}
where the following notations are used
\begin{align*}
&x(t)=\left(x_1(t), x_2(t), x_3(t)\right)^T=(\dot{\delta}_{y_2}(t), \dot{\delta}_{y_4}(t), f_3(t))^T, \quad\tau=\Delta t_3,\\[5pt]
&a_1=2\pi c_3\frac{2c_1}{c_1+c_3},\quad a_2=2\pi c_3\frac{c_1-c_3}{c_1+c_3},\quad 
a_5=\frac{2c_5}{c_5+c_3} .
\end{align*}
All functions and parameters in the system are dimensionless. Note that $a_1-a_2=2\pi c_3\not=0$ when $\theta_3\not=\pi/2$.  
It is useful to write the system \eqref{DDEf3} to a matrix form as
\begin{equation}\label{DDDE}
\left(
	\begin{array}{c}
		\dot x_1(t) \\
		\dot x_2(t) \\
                0
	\end{array}
\right)=
A x(t)+Bx(t-\tau)+Cx(t-2\tau)+h\left(t\right),\quad t\geq 0,
\end{equation}
where the constant matrices and the given source term $h:\mathbb{R}\to\mathbb{R}^3$ are, respectively, 
\begin{equation*}
A=
\left(
	\begin{array}{ccc}
		-a_1 r_2 & 0 & a_1a_0\\
		 0 & -(a_1-a_2) r_4a_5 & 0\\
                 0&0&-1\\
	\end{array}
\right),\ 
B=
\left(
	\begin{array}{ccc}
                 0&0&0\\
		-a_5 a_1 r_2 & 0 & a_2a_0a_5\\
		 0 & -\frac{a_5 r_4}{a_0} & 0\\                 
	\end{array}
\right),\ 
\end{equation*}
\small{
\begin{equation*}
  C=
  \left(
    \begin{array}{ccc}
      0&0&0\\
      0&0&0\\
      -\frac{a_1 r_2 (a_5-1)}{(a_1-a_2)a_0} & 0 & \frac{a_2(a_5-1)}{a_1-a_2}\\
    \end{array}
  \right),\
  h(t)=\left(
    \begin{array}{c}
      a_1 a_0 F(t)\\
      a_1 a_0 a_5 F(t-\tau)\\
      \frac{a_1(a_5-1)}{a_1-a_2} F(t-2\tau)\\[5pt]
    \end{array}
  \right).
\end{equation*}
}
\normalsize
The reflected wave $f_1$, and the transmitted wave $g_5$ in their dimensionless forms are
\begin{equation}\label{eq:36_refl}
f_1(t)
= \frac{1}{c_1+c_3}
\left[(c_3-c_1)F(t)
-2 c_3
\frac{r_2}{a_0}x_1 (t)
+2 c_3 x_3(t)
\right]
\end{equation}
and
\begin{align}
g_5(t)
&= \frac{2 c_1}{c_1+c_3}
\left[F(t+\Delta t_5 -\Delta t_3)
-
\frac{r_2}{a_0} x_1(t+\Delta t_5 -\Delta t_3)
\right]\nonumber\\[5pt]
&+\frac{c_1-c_3}{c_1+c_3} x_3(t+\Delta t_5 -\Delta t_3)
-x_3(t+\Delta t_5 +\Delta t_3)
-\frac{r_4}{a_0} x_2(t+\Delta t_5).
\label{eq:36_tran}
\end{align}

These have been left in their original notational form as they are calculated from the solution of the system
\eqref{DDEf3}.
\subsection{The solution of the singularly perturbed system}
This section considers the  following linear non-homogeneous system of delay differential equations 
\begin{equation}\label{PS}
\frac{d}{dt}(E(\epsilon)x(t))=A x(t)+Bx(t-\tau)+Cx(t-2\tau)+h\left(t\right),\quad t\geq 0,
\end{equation}
where the coefficient matrices $A,B,C$ and the source term $h$ are as in the HS system \eqref{DDDE}, 
$\epsilon\geq 0$ and 
\[
E(\epsilon)=\begin{pmatrix}
I&0\\
0&\epsilon
\end{pmatrix},
\]
with $I\in\mathbb{R}^{2\times 2}$ the identity matrix. We investigate the limit behavior as the small parameter $\epsilon\to 0^+$, of solutions of \eqref{PS} on  $[0, \infty)$. In \cite{cooke65, cooke66}, the authors examine conditions which guarantee that the solutions of \eqref{PS} converge, as $\epsilon\to 0^+$, to the solution of the 
differential-difference system obtained when in \eqref{PS} the value of the parameter is $\epsilon=0$.  
We show that the conditions that guarantee convergence as $\epsilon\to 0^+$ to the solution of \eqref{DDDE} are satisfied for this system and this limit is used to give the solution to the HS. 

The solution of the perturbed system \eqref{PS} is denoted by $x(\epsilon, t)$, to emphasize its dependence on the parameter $\epsilon$. When $\epsilon\not=0$, then $\det E(\epsilon)=\epsilon\not=0$ and Theorem 6.2 
in \cite{bellman} can be applied to ensure existence and uniqueness of the solution to the system \eqref{PS}, 
with given initial function $\phi\in C([-2\tau,0],\mathbb{R}^3)$ and source term $h\in C([0,\infty),\mathbb{R}^3)$. When $\epsilon=0,$ if the initial function 
is continuous and of bounded variation for $t\in[-2\tau,0],$ and $h$ is continuous and of bounded variation for $t\in[0,\infty),$ then Theorem 1 in \cite{cooke66} can be applied to ensure existence and uniqueness of the solution $x(0,t)=x(t)$ of \eqref{DDDE}. 

First, we give the solution of \eqref{PS} when $\epsilon\not=0$ by analyzing the characteristic roots. 
Assume that the initial function $\phi$ does not depend on $\epsilon$. Denote by $X,\ \Phi$, and $H$ the Laplace transforms of the corresponding functions: $X = \mathcal{L}(x),\ \Phi= \mathcal{L}(\phi(\cdot - 2\tau)),\ H = \mathcal{L}(h)$, where we have extended $\phi$ to $\left[-2\tau,\infty\right)$ 
by making it zero for $t>0$.
Applying the Laplace transform to the system \eqref{PS} results in 
\begin{equation}\label{Laplace_sol_eps}
X(\epsilon,s)
= \Delta^{-1}(\epsilon,s)\left[
E(\epsilon)\phi(0) + B\Phi(s) + C\Phi(s)+H(s)
\right],
\end{equation}
where $\Delta(\epsilon,s)$ is the characteristic matrix, defined by
\begin{equation}\label{char_matrix}
\Delta(\epsilon,s) = sE(\epsilon)-A-Be^{-\tau s}-Ce^{-2\tau s}.
\end{equation}
Taking the inverse Laplace transform of \eqref{Laplace_sol_eps}, the 
solution of the inhomogeneous problem \eqref{PS} with initial function $\phi$ is
\begin{equation}\label{solution_inverseLaplace_eps}
x(\epsilon,t)=\int_{(a)}e^{st}\Delta^{-1}(\epsilon,s)
 \left[
E(\epsilon)\phi(0) + B\Phi(s) + C\Phi(s)+H(s)
\right] ds
\end{equation}
for any sufficiently large constant $a>\sup\{\Re(s)\mid \det\Delta(\epsilon,s)=0\}$,  where 
\[
\int_{(a)}=\frac{1}{2\pi i}\lim_{T\to\infty}\int_{a-iT}^{a+iT}.
\]
The integrand in \eqref{solution_inverseLaplace_eps} is a 
meromorphic function with, possibly, poles at the roots of the characteristic equation
\begin{equation}\label{char_eq_eps}
\det\Delta(\epsilon,s)=0.
\end{equation}
Let $\tilde x(\epsilon,t)$ be the fundamental matrix solution of the homogeneous problem i.e., \eqref{PS} with $h=0$,  for $t\geq 0$ that satisfies the initial condition 
\[
\tilde x(\epsilon,t)=\begin{cases}
0, \quad & t<0,\\
E^{-1}(\epsilon), \quad & t=0.
\end{cases}
\] 
As $\Delta^{-1}(\epsilon,s)$ is the Laplace transform of $\tilde x(\epsilon,t)$ (see also \cite{bellman}, \cite{Hale}), 
\begin{equation}\label{fund_sol_eps}
\tilde x(\epsilon,t)=\int_{(a)}e^{st}\Delta^{-1}(\epsilon,s)ds,\quad a>\sup\{\Re(s)\mid \det\Delta(\epsilon,s)=0\}.
\end{equation}
By the convolution theorem, the solution $x(\epsilon,t)$ in \eqref{solution_inverseLaplace_eps} can be obtained as 
\begin{align}\label{solution_inverseLaplace1_eps}
x(\epsilon,t)= &\, \tilde x(\epsilon,t)E(\epsilon)\phi(0) + \int_0^\tau \tilde x(\epsilon,t-\theta) B\phi(\theta-\tau)d\theta\nonumber\\ 
&+ \int_0^{2\tau} \tilde x(\epsilon,t-\theta) C\phi(\theta-2\tau)d\theta 
+\int_0^t\tilde x(\epsilon,t-\theta)h(\theta)d\theta,\quad t\geq 0.
\end{align}
Information on the roots of the characteristic equation \eqref{char_eq_eps}, i.e., 
the elements of the spectrum
 \[
\sigma(\epsilon)=\{\lambda\in\mathbb{C}\mid \det\Delta(\epsilon,\lambda)=0\}.
\]
is needed in order to study the asymptotic behavior of the solution of \eqref{PS} as $t\to\infty$. 

The following lemma gives all the information about the location of these roots.
\begin{lemma}\label{l:spectrumPS}
  For the roots of the characteristic equation \eqref{char_eq_eps}, the followings hold for all $\epsilon\geq 0:$
  \begin{itemize}
  \item[(a)] $\lambda=0\in \sigma(\epsilon)$,  and it is a simple root. 
  \item[(b)] $\forall \lambda\in\sigma(\epsilon)\setminus\{0\}$,  $\Re(\lambda)< 0$. 
  \end{itemize}
\end{lemma}
\begin{proof}
The characteristic equation $\det\Delta(\epsilon,\lambda)=0$ is equivalent to
\begin{align}\label{char_eq_eps1}
(1+\lambda\epsilon)&\left(\lambda+a_5r_4(a_1-a_2)\right)\left(\lambda+a_1 r_2\right)\nonumber\\
&=e^{-2\tau \lambda}\left( (\lambda+a_1r_2)a_2-a_1^2 r_2\right)
\left(\lambda\frac{a_5-1}{a_1-a_2}-a_5r_4\right).
\end{align}
Observe that $\lambda=0$ is a simple root for any $\epsilon\geq 0$. 
The proof of part (b) consists of two steps. It is initially shown that the assertion is true for $\epsilon= 0$,  
 and then, using this result, it is shown that it must also hold for any $\epsilon>0$.

\noindent{\em Step 1.} When $\epsilon= 0$,  then $\det\Delta(0,\lambda)=0$ is equivalent to
\begin{equation}\label{char_eq_all}
\left(\lambda+a_5r_4(a_1-a_2)\right)\left(\lambda+a_1 r_2\right)=e^{-2\tau \lambda}\left( (\lambda+a_1r_2)a_2-a_1^2 r_2\right)
\left(\lambda\frac{a_5-1}{a_1-a_2}-a_5r_4\right).
\end{equation}
Let $\lambda=x+iy$ be a root and be substituted into 
\eqref{char_eq_all}, then the absolute value square is taken for both sides to obtain
\begin{align}\label{char_eq_all1}
&\left( (x+a_1r_2)^2+y^2\right)\left( (x+2\pi c_3a_5r_4)^2+y^2\right)\nonumber\\
&=\frac{e^{-4\tau x}}{4\pi^2c_3^2}
\left[ \left(a_2x+a_1r_2(a_1-a_2)\right)^2+a_2^2y^2\right]
\left[\left( (a_5-1)x-2\pi c_3 a_5r_4\right)^2+(a_5-1)^2y^2 \right].
\end{align}  
For any $y\in\mathbb{R}$,   the left and the right hand sides of \eqref{char_eq_all1} are denoted by $l(x)$ and $r(x)$, respectively. Then
\[
l'(x)=2\left(2x+a_1r_2+2\pi c_3a_5r_4\right)\left( (x+a_1r_2)(x+2\pi c_3 a_5r_4)+y^2
\right),
\] 
which has roots
\[
x_0=-\frac{a_1r_2+2\pi c_3a_5r_4}{2},\ x_\pm=\frac{-(a_1r_2+2\pi c_3a_5r_4)\pm 
\sqrt{(a_1r_2-2\pi c_3a_5r_4)^2-4y^2}}{2}.
\]
If $y$ is such that $x_\pm$ are real, then $x_-\leq x_0\leq x_+<0$ holds, and at $x_+$ the function $l$ has a 
local minimum. When $y$ is large enough, such that $x_\pm$ are complex, then $l'(x)=0$ has only one solution $x_0<0$,  which is 
a global minimum point for $l$.  In both cases $l(x)$ is strictly increasing for $x>0$,  and 
\[
l(0)=(a_1^2r_2^2+y^2)(4\pi^2c_3^2a_5^2r_4^2+y^2).
\]
It can be shown that the function $r$ is strictly decreasing, and 
\[
r(0)=\left(a_1^2r_2^2+\frac{a_2^2}{4\pi^2c_3^2}y^2\right)\left(4\pi^2c_3^2a_5^2r_4^2+(a_5-1)^2y^2\right).
\] 
Finally, consider
\begin{align*}
&l(0)-r(0)=y^2\left[\left(1-\frac{a_2^2(a_5-1)^2}{4\pi^2c_3^2}\right)y^2+a_5^2r_4^2a_1(a_1-2a_2)+a_1^2r_2^2a_5(2-a_5)\right]\\
&=y^2\left[\frac{a_1(a_1-2a_2)+a_2^2a_5(2-a_5)}{4\pi^2c_3^2}y^2+a_5^2r_4^2a_1(a_1-2a_2)+a_1^2r_2^2a_5(2-a_5)\right].
\end{align*}
Since $a_1>0$,  $a_5>0$, 
\[
2-a_5=\frac{2c_3}{c_5+c_3}>0,\quad a_1-2a_2=\frac{4\pi^2c_3^2}{c_1+c_3}>0,
\]
it follows that $l(0)-r(0)>0$ for all $y\not=0$. 

Combining all results show that $l$ is increasing and $r$ is decreasing for $x>0$ and $l(0)>r(0)$,  
hence $l(x)>r(x)$ for all $x>0$.  Consequently, the curves of $l$ and $r$ can intersect only at $x<0$. 

\noindent{\em Step 2.} Let $\epsilon>0$ and $\lambda=x+iy$ be a root of $\det\Delta(\epsilon,\lambda)=0$.  
Take for both sides in \eqref{char_eq_eps1} the absolute value square and denote the left hand side by 
$l(\epsilon,x)$.  The right hand side remains $r(x)$ since it does not depend on $\epsilon$.  Let
\begin{equation}
d(\epsilon,x)=l(\epsilon,x)-r(x)=\left((1+\epsilon x)^2+\epsilon^2 y^2\right) l(x)-r(x),
\end{equation} 
where $l(x)=l(0,x)$ is as before.  It is known from Step 1 that $l(x)>0$ and $d(0,x)>0$ for all $x>0$.  Then for any $x>0$,  
\[
\frac{d}{d\epsilon}d(\epsilon,x)= 2\left((1+\epsilon x)x+\epsilon y^2\right)l(x)>0.
\]
Hence, $d$ is a strictly increasing function of $\epsilon$,  that is, for any $x>0$,  
$d(\epsilon,x)>d(0,x)>0$ holds for all $\epsilon>0$.  Moreover,
$l(\epsilon,0)=(1+\epsilon^2 y^2)l(0)>l(0)>r(0)$.  It follows that for
any $\epsilon\geq 0$,  $l(\epsilon,x)=r(x)$ is only possible for $x<0$,  which completes the proof of the lemma.
\end{proof}
To study the dynamic behavior of the fundamental matrix solution $\tilde{x}(\epsilon,t)$,  the line of integration 
in \eqref{fund_sol_eps} is shifted to the left, whilst keeping track of the residues  corresponding to the singularities of $\Delta^{-1}(\epsilon,s)$ that are passed. 
The Cauchy theorem of residues implies that (see \cite{DiekmannRFDE}, \cite{Hale})
\begin{equation}\label{solution_inverseLaplaceRes_eps}
\tilde x(\epsilon,t)=\int_{(\alpha_m)}e^{st}\Delta^{-1}(\epsilon,s)ds + \sum_{j=1}^{k_m}\underset{s=\lambda_j}{\text{Res}}\,e^{st}\Delta^{-1}(\epsilon,s),
\end{equation}
where $\lambda_1,\dots,\lambda_{k_m}$ are roots of the characteristic equation, such that $\Re(\lambda_j)>\alpha_m$ 
$\forall j=1,\dots, k_m$.  
Moreover, if $\lambda_j$ is a zero of $\det\Delta(\epsilon,s)=0$ of order m, then
\begin{equation}\label{pj-polynomials_eps}
\underset{s=\lambda_j}{\text{Res}}\,e^{st}\Delta^{-1}(\epsilon,s)=p_j(\epsilon,t)e^{\lambda_jt},
\end{equation}
where $p_j$ is a $\mathbb{C}^{3\times 3}$-valued polynomial in $t$ of degree less than or equal to $m-1$,  with 
coefficients that depend on $\epsilon$.  

Using Lemma \ref{l:spectrumPS}, $\alpha_m=-\alpha<0$ can be chosen in \eqref{solution_inverseLaplaceRes_eps} so that only the $\lambda=0$ (simple) root of the characteristic equation is to the right of the line $\{s\mid \Re(s)=-\alpha\}$. 
Hence, we can conclude that the matrix solution of the homogeneous system in \eqref{solution_inverseLaplaceRes_eps} is
\begin{equation}\label{solution_eps}
\tilde x(\epsilon,t)=\int_{(-\alpha)}e^{st}\Delta^{-1}(\epsilon,s)ds+\underset{s=0}{\text{Res}}\,e^{st}\Delta^{-1}(\epsilon,s)=
\int_{(-\alpha)}e^{st}\Delta^{-1}(\epsilon,s)ds+M(\epsilon), 
\end{equation}
where $M(\epsilon)$ is a constant matrix.
Moreover, as the 
residue of $e^{st}\Delta^{-1}(\epsilon,s)$ at any root of the characteristic equation is a solution of the 
homogeneous system, it follows that $M(\epsilon)$ is a matrix solution. Consequently,
\begin{equation}
(A+B+C)M(\epsilon)=0,
\end{equation} 
which determines 
\begin{equation}\label{Mepsilon}
M(\epsilon)=\begin{pmatrix}
m_{1} & -\frac{r_4}{r_2}m_{2} & \frac{a_0}{r_2}m_{3}\\[5pt]
-\frac{r_2}{r_4}m_{1} & m_{2} & -\frac{a_0}{r_4}m_{3}\\[5pt]
\frac{r_2}{a_0}m_{1} & -\frac{r_4}{a_0}m_{2} & m_{3}
\end{pmatrix},
\end{equation}
with $m_{1},m_{2}$ and $m_{3}$ non-zero real parameters that can depend on $\epsilon$.

Furthermore, denoting the integral in \eqref{solution_eps} by $N(\epsilon,t)$, the estimate
\begin{equation}\label{estimateN_eps}
\|N(\epsilon,t)\|=\|\int_{(-\alpha)}e^{st}\Delta^{-1}(\epsilon,s)ds\|\leq K(\epsilon)e^{-\alpha t}
\end{equation}
holds for some $K(\epsilon)>0$,  where $\alpha>0$ and $\|\cdot\|$ is some matrix norm.

As we are interested in the asymptotic behavior of the solutions as $t\to\infty,$ our main result is summarized in the following lemma.

\begin{lemma}\label{l:asymptotic_behavior_eps}
Suppose $h:[0,\infty)\to \mathbb{R}^3$ is a given exponentially bounded function, i.e., there are $K_1>0$,  $\beta>0$ constants 
such that
\[
\|h(t)\|\leq K_1 e^{-\beta t},\ t\geq 0.
\]
The asymptotic behavior of the solution $x(\epsilon, t)$ of \eqref{PS} for $\epsilon>0$, with given initial function $\phi\in C([-2\tau,0])$ as $t\to\infty$ is then
\begin{align}\label{solution_limit_eps}
\lim_{t\to\infty}x(\epsilon,t)=&M(\epsilon)\left[E(\epsilon)\phi(0) + B\int_0^\tau \phi(\theta-\tau)d\theta  
+ C\int_0^{2\tau} \phi(\theta-2\tau)d\theta\right] \nonumber\\
&+M(\epsilon)\int_0^\infty h(\theta)d\theta,
\end{align}
where $M(\epsilon)$ is the matrix in \eqref{Mepsilon}.
\end{lemma}
\begin{proof}
Introduce \eqref{solution_eps} into \eqref{solution_inverseLaplace1_eps} to obtain
\begin{align}\label{solution_inverseLaplace2}
&x(\epsilon,t)=(N(\epsilon,t)+M(\epsilon))E(\epsilon)\phi(0) + \int_0^\tau (N(\epsilon,t-\theta)+M(\epsilon)) B\phi(\theta-\tau)d\theta  \nonumber\\
&+\int_0^{2\tau} (N(\epsilon,t-\theta)+M(\epsilon)) C\phi(\theta-2\tau)d\theta
+\int_0^t (N(\epsilon,t-\theta)+M(\epsilon))h(\theta)d\theta.
\end{align}
Using the estimate \eqref{estimateN_eps} and since $\alpha>0$,  
\begin{equation}
\lim_{t\to\infty}\tilde x(\epsilon,t)=\lim_{t\to\infty}\left(M(\epsilon)+N(\epsilon,t)\right)= M(\epsilon).
\end{equation}
Take the limit $t\to\infty$ in \eqref{solution_inverseLaplace2} to obtain
\begin{align}\label{solution_inverseLaplace3}
\lim_{t\to\infty}x(\epsilon,t)=&M(\epsilon)E(\epsilon)\phi(0) + \int_0^\tau M(\epsilon) B\phi(\theta-\tau)d\theta  \nonumber\\
&+ \int_0^{2\tau} M(\epsilon) C\phi(\theta-2\tau)d\theta
+\lim_{t\to\infty}\int_0^t (N(\epsilon,t-\theta)+M(\epsilon))h(\theta)d\theta.
\end{align}
Using the assumption on the incoming inhomogeneity results in
\begin{align*}
\|&\int_0^t N(\epsilon,t-\theta)h(\theta)d\theta\|\leq K K_1\int_0^t e^{-\alpha (t-\theta)}e^{-\beta\theta}d\theta 
= K K_1e^{-\alpha t}\int_0^t e^{(\alpha-\beta)\theta}d\theta\\[5pt]
&=\frac{K K_1}{\alpha-\beta}e^{-\alpha t}\left(  e^{(\alpha-\beta)t}-1\right)
=\frac{K K_1}{\alpha-\beta}\left(  e^{-\beta t}-e^{-\alpha t}\right)\to 0,\quad \text{as } t\to\infty.
\end{align*}
This is then combined with \eqref{solution_inverseLaplace3} yielding the desired result.
\end{proof}
The next goal is to give explicit formulas for the coefficients of the polynomials $p_j(\epsilon,t)$ 
in \eqref{pj-polynomials_eps}, as linear operators acting on the initial function $\phi$,  in terms of 
the spectral information. According to Lemma 3.6. in \cite{DiekmannRFDE}, if 
\[
\Delta^{-1}(\epsilon,s)=\frac{G(\epsilon,s)}{s-\lambda},\quad \text{with } G \text{ analytic at } s=\lambda,
\]
where $\lambda$ is a zero of $\det\Delta (\epsilon,s)$,  then the generalized eigenspace $\mathcal{M}_\lambda(\epsilon)$ 
at $\lambda$ is given by
\[
\mathcal{M}_\lambda(\epsilon)=\{\theta\mapsto e^{\lambda\theta}v\mid v\in\mathcal{N}\left(\Delta(\epsilon,\lambda)\right)\}, 
\] 
with $\mathcal{N}\left(\Delta(\epsilon,\lambda)\right)$ denoting the null space of the operator $\Delta(\epsilon,\lambda)$.  Moreover, 
the projection $\mathcal{P}_\lambda$ of $\phi\in C([-2\tau,0])$ onto $\mathcal{M}_\lambda(\epsilon)$ is given by
\begin{equation}
\mathcal{P}_\lambda \phi =e^{\lambda\cdot}G(\epsilon,\lambda)\left[E(\epsilon)\phi(0)
+B e^{-\lambda\tau} \int_{-\tau}^0 e^{-\lambda\theta}\phi(\theta)d\theta +C e^{-2\lambda\tau}\int_{-2\tau}^0 e^{-\lambda\theta}\phi(\theta)d\theta\right].\label{projection}
\end{equation}
Having the spectral information contained in Lemma \ref{l:spectrumPS}, the constant matrix $M(\epsilon)$ is determined  in \eqref{solution_eps} by applying the results above to the $\lambda=0$ simple eigenvalue. 
As
\begin{align*}
\Delta^{-1}(\epsilon,s)&=\frac{1}{\det(\Delta(\epsilon,s))}\text{adj }\Delta(\epsilon,s),
\end{align*}
it is straightforward to see that there exists a $G(\epsilon,s)$,  analytic function at $s=0$,  and is given by
\begin{align*}
G(\epsilon,s)=\frac{e^{2\tau s}}{g(\epsilon,s)}\text{adj }\Delta(\epsilon,s),
\end{align*}
with $g(\epsilon,s)$ and $\text{adj }\Delta(\epsilon,s)$ given explicitly in \ref{app:parameters}.
The projection $\mathcal{P}_0$ of $\phi$ onto the one dimensional eigenspace $\mathcal{M}_0(\epsilon)$, according to \eqref{projection}, is
\begin{equation}
\mathcal{P}_0 \phi =G(\epsilon,0)\left[E(\epsilon)\phi(0)
+B \int_{-\tau}^0 \phi(\theta)d\theta +C \int_{-2\tau}^0 \phi(\theta)d\theta\right],
\end{equation}
where
\begin{equation}\label{H_eps}
G(\epsilon,0)=\frac{a_1 a_5 r_4}{g(\epsilon,0)}
\begin{pmatrix}
1  & -1 & a_0(a_1-a_2)\\
 -\frac{r_2}{r_4} & \frac{r_2}{r_4} & -\frac{r_2}{r_4} a_0(a_1-a_2)\\
\frac{r_2}{a_0} & -\frac{r_2}{a_0} & r_2(a_1-a_2)
\end{pmatrix}, 
\end{equation}
and
\[
g(\epsilon,0)= a_1 a_5[r_2+r_4+2\tau r_2r_4(a_1-a_2)+\epsilon r_2r_4(a_1-a_2)].
\]
Consequently, the constant matrix $M(\epsilon)$ in Lemma \ref{l:asymptotic_behavior_eps} is $M(\epsilon)=G(\epsilon,0)$,  i.e.,
\[
m_1=\frac{a_1a_5r_4}{g(\epsilon,0)},\quad m_2=\frac{a_1a_5r_2}{g(\epsilon,0)},\quad m_3=\frac{a_1a_5r_4}{g(\epsilon,0)}r_2(a_1-a_2).
\]
Note that, in general, when the initial function $\phi$ is the eigenvector  corresponding to the eigenvalue $\lambda$, then $\mathcal{P}_\lambda \phi=\phi$.  This is shown in next remark  for the zero eigenvalue and its corresponding eigenvector.
\begin{remark}\label{r:Pphi}
In the case of constant initial functions $\phi(\theta)=\phi(0)=(\phi_1,\phi_2,\phi_3)^T$,  $-2\tau\leq\theta\leq 0$, the limit in \eqref{solution_limit_eps} becomes
\begin{align}\label{Peigenvector_eps}
\lim_{t\to\infty}x(\epsilon,t)&=\mathcal{P}_0\phi+M(\epsilon)\int_0^\infty h(\theta)d\theta\nonumber\\
&= M(\epsilon)\left[E(\epsilon)+\tau B+2\tau C\right]\phi(0)+M(\epsilon)\int_0^\infty h(\theta)d\theta\nonumber\\
&=c_\phi(\epsilon)
\begin{pmatrix}
1\\
-\frac{r_2}{r_4}\\
\frac{r_2}{a_0}
\end{pmatrix}+M(\epsilon)\int_0^\infty h(\theta)d\theta,
\end{align}
where the multiplication factor $c_\phi(\epsilon)$ is
\begin{equation}\label{cphi_eps}
c_\phi(\epsilon)=\frac{r_4(\phi_1d_1+\phi_2 d_2+\phi_3d_3)}{r_2+r_4+r_2r_4(a_1-a_2)2\tau+r_2r_4(a_1-a_2)\epsilon}\ ,
\end{equation}
with $d_1=1+2\tau a_1r_2 -\tau a_1a_5r_2$,  $d_2=-1-a_5r_4\tau(a_1-a_2)$,  $d_3=a_0a_2\tau(a_5-2)+a_0(a_1-a_2)\epsilon$.  
In the special case when $\phi$ is an eigenvector corresponding to the zero eigenvalue, i.e., 
$\phi=(1,-r_2/r_4,r_2/a_0)^T\in\mathcal{N}(\Delta(\epsilon,0))$,  then $c_\phi(\epsilon)=1$ for all 
$\epsilon>0$, hence $\mathcal{P}_0 \phi=\phi$. 
\end{remark}
Summarizing, for any $\epsilon>0$ the solution of the DDE system \eqref{PS} is given as \eqref{solution_inverseLaplace2}. In this expression, we determined the matrix $M(\epsilon)$ and gave an exponentially decaying bound for $N(\epsilon, t)$. 
Next, we look at the limiting behavior of the solutions as $\epsilon\to 0^+.$ 
The following lemma is a direct application of the Convergence Theorem in \cite{cooke66}.
\begin{lemma}
Let $x(\epsilon, t)$ be the solution of \eqref{PS} corresponding to the initial function $\phi(t)$,  which is continuous and 
of bounded variation on $[-2\tau,0]$,  and where $h$ is continuous and of bounded variation on $[0,\infty)$.  Then 
\begin{equation}\label{limitx1x2}
\lim_{\epsilon\to 0+}x_1(\epsilon, t)=x_1(0,t),\ \lim_{\epsilon\to 0+}x_2(\epsilon, t)=x_2(0,t),
\end{equation}
where the convergence is uniform in $t$ for $t\in[0,\infty)$.  If $\dot\phi(t)$ exists, it is continuous and 
of bounded variation on $[-2\tau,0]$,  and if $\dot h(t)$ exists, it is continuous and 
of bounded variation on $[0,\infty)$ then
\begin{equation}\label{limitx3}
\lim_{\epsilon\to 0+}x_3(\epsilon, t)=x_3(0,t), \text{ for all } t\in CQ,
\end{equation}
where $CQ=[0,\infty)\setminus\{t^* \mid t^*=j\tau, j\geq 0, j\in\mathbb{Z}\}$. The convergence in 
\eqref{limitx3} is uniform in $t$ for any compact subset of $CQ$.  If moreover, 
\[
\frac{d}{dt}(E(\epsilon)\phi(t))\mid_{t=0^-}=A \phi(0)+B\phi(-\tau)+C\phi(-2\tau)+h(0)
\]
holds, then the convergence in \eqref{limitx3} will be uniform for $t$ in $[0,\infty)$. 
\end{lemma}
\begin{proof}
We apply the Convergence Theorem in \cite{cooke66}. It can be verified that the condition of this theorem on the $\sigma_0$-complete regularity, for some $\sigma_0$, is satisfied 
as the element $A(3,3)=-1$ of the matrix $A$ in \eqref{PS} is negative . It then follows  from this theorem that the limits  in \eqref{limitx1x2} hold and the convergence is uniform in $t$ for any bounded subset of $[0,\infty)$.
Furthermore, by using \eqref{solution_inverseLaplace2} and the spectral properties of the system in 
Lemma~\ref{l:spectrumPS}, it is straightforward to show that the convergence in \eqref{limitx1x2} is uniform in $t$ for $t\in[0,\infty)$.  
\end{proof}
\subsection{The special case $n_3=n_5$}
In this section, the special case when the second metal layer is embedded in two dielectrics that have 
the same index of refraction, i.e., $n_3=n_5$ is considered.
 In this case $a_5-1=0$,  hence \eqref{x3} decouples from \eqref{x1} and \eqref{x2}, in the
sense that
\[
x_3(t)=-\frac{r_4}{a_0} x_2(t-\tau).
\]
Due to this decoupling, there is no need to use singular perturbation, as the hybrid system (\ref{DDEf3}) 
reduces to a delay differential system with two delays, $\tau$ and $2\tau,$
\begin{equation}\label{DDEsystem_2delays}
\dot x(t)
= A x(t)
+ B x(t-\tau) +  C x(t-2\tau)+h(t), \quad t\geq 0,
\end{equation}
where
\begin{align}\label{DDEsystem_2delays_coeff}
x(t)&=
\left(
	\begin{array}{c}
		x_1(t) \\
		x_2(t) \\
	\end{array}
\right),\
h(t)=
a_1 a_0\left(
	\begin{array}{c}
		F(t) \\
		F(t-\tau) \\
	\end{array}
\right),\ 
C=
\left(
	\begin{array}{cc}
		  0 & 0\\
		 0 & -a_2 r_4 \\
	\end{array}
\right),
\nonumber\\[9pt]
A&=
\left(
	\begin{array}{cc}
		-a_1 r_2 & 0 \\
		 0 & -(a_1-a_2) r_4\\
	\end{array}
\right),\
B=
\left(
	\begin{array}{cc}
		  0 & -a_1 r_4\\
		 -a_1 r_2 & 0 \\
	\end{array}
\right).
\end{align}
In this special case, the reflected and transmitted waves $f_1(t)$ and $g_5(t)$ are
\begin{equation}\label{eq:36_refl}
f_1(t)
= \frac{1}{c_1+c_3}
\left[(c_3-c_1)F(t)
-2 c_3
\frac{r_2}{a_0}x_1 (t)
-2 c_3 \frac{r_4}{a_0}x_2(t-\Delta t_3)
\right],
\end{equation}
\begin{align}
g_5(t)
= \frac{2 c_1}{c_1+c_3}
\left[F(t)
-
\frac{r_2}{a_0} x_1(t)
\right]
-\frac{c_1-c_3}{c_1+c_3} \frac{r_4}{a_0}x_2(t-\Delta t_3).
\label{eq:36_tran}
\end{align}
Applying the Laplace transform to the system \eqref{DDEsystem_2delays} results in
\begin{equation}\label{eq:36_an_2}
X(s)
= \Delta^{-1}(s)\left[
\phi(0) + B\Phi(s) + C\Phi(s)+H(s)
\right],
\end{equation}
where $\Delta(s)$ is the characteristic matrix, defined by
\begin{equation}\label{char_matrix}
\Delta(s) = sI-A-Be^{-\tau s}-Ce^{-2\tau s},
\end{equation}
with $I$ the identity matrix, and where $\phi\in C([-2\tau,0],\mathbb{R}^2)$ is the initial function. Taking the inverse Laplace transform of \eqref{eq:36_an_2}, the 
solution of the inhomogeneous initial value problem  \eqref{DDEsystem_2delays} is
\begin{equation}\label{solution_inverseLaplace}
x(t)=\int_{(a)}e^{st}\Delta^{-1}(s)
 \left[
\phi(0) + B\Phi(s) + C\Phi(s)+H(s)
\right] ds,
\end{equation}
for any sufficiently large constant $a>\sup\{\Re(s)\mid \det\Delta(s)=0\}$.  
If $\tilde x(t)$ is the fundamental matrix solution of the homogeneous problem, then the solution $x(t)=x(t;\phi,h)$ in 
\eqref{solution_inverseLaplace} is 
\begin{align}\label{solution_inverseLaplace1}
x(t)=&\tilde x(t)\phi(0) + \int_0^\tau \tilde x(t-\theta) B\phi(\theta-\tau)d\theta 
+ \int_0^{2\tau} \tilde x(t-\theta) C\phi(\theta-2\tau)d\theta \nonumber\\
&+\int_0^t\tilde x(t-\theta)h(\theta)d\theta,\quad t\geq 0.
\end{align}

When studying the asymptotic behavior of the solution of the system \eqref{DDEsystem_2delays}, as $t\to\infty$,  the location of the roots of the characteristic equation, i.e., the elements of 
$\sigma=\{\lambda\in\mathbb{C}\mid \det\Delta(\lambda)=0\}$ must be known. 
\begin{lemma}\label{l:spectrum}
For the roots of the characteristic equation, the followings hold
\begin{itemize}
\item[(a)] $\lambda=0\in \sigma$ and it is a simple root. 
\item[(b)] $\forall \lambda\in\sigma\setminus\{0\}$,  $\Re(\lambda)< 0$. 
\end{itemize}
\end{lemma}
\begin{proof}
The proof of this lemma can be found in  \ref{app:proof_lemma}.
\end{proof}
Using Lemma \ref{l:spectrum}, the fundamental matrix 
solution of the system is
\begin{equation}\label{solution}
\tilde x(t)=\int_{(-\alpha)}e^{st}\Delta^{-1}(s)ds+\underset{s=0}{\text{Res}}\,e^{st}\Delta^{-1}(s)=
\int_{(-\alpha)}e^{st}\Delta^{-1}(s)ds+M, 
\end{equation}
where $M$ is a constant matrix. Moreover, 
\begin{equation}
(A+B+C)M=0,
\end{equation} 
which determines 
\[
M=\begin{pmatrix}
m_{1} & -\frac{r_4}{r_2}m_{2}\\[5pt]
-\frac{r_2}{r_4}m_{1} & m_{2}
\end{pmatrix},
\]
with $m_{1}$ and $m_{2}$ non-zero real parameters. Furthermore, denoting the integral in \eqref{solution} by $N(t)$,  the estimate
\begin{equation}\label{estimateN}
\|N(t)\|=\|\int_{(-\alpha)}e^{st}\Delta^{-1}(s)ds\|\leq Ke^{-\alpha t}
\end{equation}
holds for some $K>0$,  where $\alpha>0$.  

The following corollary summarizes the main results and the proof is completely identical to the proof of 
Lemma \ref{l:asymptotic_behavior_eps}.
\begin{corollary}\label{l:asymptotic_behavior}
Suppose $h:[0,\infty)\to \mathbb{R}^2$ is a given exponentially bounded function, i.e., there are $K_1>0$,  $\beta>0$ constants 
such that
\[
\|h(t)\|\leq K_1 e^{-\beta t},\ t\geq 0.
\]
Then the asymptotic behavior of the solution $x(t)=x(t;\phi,h)$ of  \eqref{DDEsystem_2delays}, with given initial function 
$\phi\in C([-2\tau,0])$,  as $t\to\infty$ is
\begin{equation}\label{solution_limit}
\lim_{t\to\infty}x(t)=M\left[\phi(0) + B\int_0^\tau \phi(\theta-\tau)d\theta  
+ C\int_0^{2\tau} \phi(\theta-2\tau)d\theta+\int_0^\infty h(\theta)d\theta\right].
\end{equation}
\end{corollary}
The projection $\mathcal{P}_0$ of $\phi$ onto the one dimensional eigenspace $\mathcal{M}_0$ is
\begin{equation}
\mathcal{P}_0 \phi =G(0)\left[\phi(0)
+B \int_{-\tau}^0 \phi(\theta)d\theta +C \int_{-2\tau}^0 \phi(\theta)d\theta\right],
\end{equation}
where
\begin{equation}\label{g0}
G(0)=\frac{a_1 r_4}{g(0)}
\begin{pmatrix}
1  & -1\\
 -\frac{r_2}{r_4} & \frac{r_2}{r_4}
\end{pmatrix}, \quad
g(0)= a_1(r_2+r_4)+2\tau a_1r_2r_4(a_1-a_2).
\end{equation}
Consequently, the constant matrix $M$ in Corollary \ref{l:asymptotic_behavior} is $M=G(0)$,  i.e.,
\[
m_1=\frac{a_1r_4}{g(0)},\quad m_2=\frac{a_1r_2}{g(0)}.
\]
In the next section,  the long time behavior of the solution for several initial functions is demonstrated using some examples. An analogue to Remark~\ref{r:Pphi} is obtained when the initial function is an eigenvector.
\begin{remark}\label{r:P_id}
In case of constant initial functions $\phi(\theta)=\phi(0)=(\phi_1,\phi_2)^T$,  $-2\tau\leq\theta\leq 0$,  the limit in 
\eqref{solution_limit} becomes
\begin{align}\label{Peigenvector}
\lim_{t\to\infty}x(t)&=\mathcal{P}_0\phi+M\int_0^\infty h(\theta)d\theta\nonumber\\
&= M\left[I+\tau B+2\tau C\right]\phi(0)+M\int_0^\infty h(\theta)d\theta\nonumber\\
&=\frac{r_4(\phi_1d_1+\phi_2 d_2)}{r_2+r_4+r_2r_4(a_1-a_2)2\tau}
\begin{pmatrix}
1\\
-\frac{r_2}{r_4}
\end{pmatrix}+M\int_0^\infty h(\theta)d\theta,
\end{align}
where $d_1=1+a_1r_2\tau$,  $d_2=-1-a_1r_4\tau+ 2\tau a_2r_4$.  Denote the multiplication factor in \eqref{Peigenvector} by
\begin{equation}\label{mf}
c_\phi=\frac{r_4(\phi_1d_1+\phi_2 d_2)}{r_2+r_4+r_2r_4(a_1-a_2)2\tau}.
\end{equation}
In the special case when $\phi$ is an eigenvector corresponding to the zero eigenvalue, i.e., 
$\phi=(1,-r_2/r_4)^T\in\mathcal{N}(\Delta(0))$ then $c_\phi=1$, hence $\mathcal{P}_0 \phi=\phi$.
\end{remark}

Note that, when $n_1=n_3=n_5$, then $a_2=0$ and thus there is only one delay in the system
\begin{equation}\label{DDEoneDelay}
\dot x(t)
= A x(t)
+ B x(t-\tau)+h(t), \quad t\geq 0,
\end{equation}
where $A, B$ and $h$ are as in \eqref{DDEsystem_2delays_coeff} (with  $a_2=0$). An analogue of Lemma~\ref{l:spectrum} 
can also be shown for this case.   
\section{Time simulations}\label{s:simulations}
In order to demonstrate the size of the parameters that enter into our 
analysis, we consider some illustrative examples, see \cite{Varro2004,Varro2007,Varro2007carrier}. 
The numerical solution of the DDE system is obtained using the Matlab \texttt{dde23} solver. 
For the numerical solution of the HS system, our own solver was implemented in combination with \texttt{dde23}. 
The system parameters are split into three groups: the incoming laser parameters, fixed for 
all simulations; the parameters that determine the size of the delay and the parameters related to the 
embedded metal layers. 

An impinging laser pulse with a Gaussian envelope is assumed
\begin{equation}
F(t)=F_0 e^{-t^2/2\tau^2_0}\cos(\omega_0 t+\phi_0),
\end{equation}
of amplitude $F_0/(V/cm)=27.46\times \sqrt{I_0/(W/cm^2)}$, with intensity $I_0=5.49\times 10^{13}\, W/cm^2$.  
The carrier frequency is $\omega_0=2.36\times 10^{15}s^{-1}$,  which corresponds to a central wavelength 
$\lambda_0=798$ nm for a Ti:Sa laser.
 The given constant pulse duration $\tau_0$ in the envelope function was taken $\tau_0=2T$, which approximately corresponds to a two-cycle pulse. More precisely, for the temporal full-with-at-half-maximum (FWHM) of the pulses we have taken $2\sqrt{\log 2}\cdot 2T$. In most simulations the carrier-envelope (CE) phase difference $\phi_0$ is considered to be zero, but its effects on the dynamics of the system are also briefly discussed. 
These parameters give a dimensionless vector potential 
$a_0=5\times 10^{-3}$,  see \ref{s:dimensionless}.

The size of the time delay $\tau=\Delta t_3$ in \eqref{delays} depends on the distance $h$ between the 
two metal layers (set to $h=10\lambda_0$ in all examples), the indices of refraction $n_1,n_3$ and the angle of incidence $\theta_1.$ Three sets of examples are considered. First, when 
$n_1=1$, $n_3=n_5=1.5$ (glass) and the angle of incidence is 
$\theta_1=\pi/3$. Then the delay time $\tau= 12.25$ optical periods. Second, when a dielectric with 
refractive index $n_3=1.1$ is inserted and $n_1$, $n_5$ and $\theta_1$ are unchanged. These result in 
a delay time $\tau= 6.78$ optical periods. Finally, when the parameters are set such that we can observe the dependence of the reflected flux on the CE phase difference of the incoming few-cycle laser 
pulse. In the first case, the DDE system~\eqref{DDEsystem_2delays} is solved. The second scenario requires 
the solution of the HS system \eqref{eq:36_5}--\eqref{eq:f3_dimless}. Moreover, here $\Delta t_5= 12.25$ optical periods, which is used in the computation of the reflected and transmitted waves.

The thicknesses $l_2$ and $l_4$ and the material of the two layers (characterized by the electron densities ${n_{e}}_2$ 
and ${n_{e}}_4$, respectively) are allowed to be different, with the assumption that the thicknesses are much smaller 
than the skin depth of the incoming radiation field, defined as $\delta_{skin}=c/\sqrt{\omega_{p}^2-\omega_0^2}$. The thickness of the two metal layers were taken first to be the same, $l_2=l_4=\lambda_0/400=2$ nm 
which, corresponding to our assumption, is smaller than $\delta_{skin2}=\delta_{skin4}=1.6\times \lambda_0/100$, when 
the plasma frequency is $\omega_{p_2}=\omega_{p_4}=10\omega_0$ (corresponding to an electron density ${n_{e}}_2={n_{e}}_4=1.8\times 10^{23}/cm^3$). The damping parameters $\Gamma_2, \Gamma_4$ in \eqref{damping2} 
were calculated from the parameters above, $\Gamma_2=\Gamma_4=0.78 \omega_0$.  
The effects of setting thinner layers were also observed.

The  displacements $\delta_{y_2}$ and $\delta_{y_4}$ are then calculated by numerically integrating the velocities $x_1$ and $x_2$,  respectively. 
\\

\noindent {\em Case 1.} Let $n_1=1$ and $n_3=n_5=1.5$. As in Remark \ref{r:P_id}, when the initial function $\phi$ is constant, then
\begin{equation}\label{lim1}
\lim_{t\to\infty}x_1(t)=c_\phi+\frac{a_1^2a_0}{g(0)}r_4\int_{-\tau}^0F(\theta)d\theta,\quad \lim_{t\to\infty}x_2(t)=-\frac{r_2}{r_4}c_\phi-\frac{a_1^2a_0}{g(0)}r_2\int_{-\tau}^0F(\theta)d\theta, 
\end{equation}
with $c_\phi$ given in \eqref{mf} and $g(0)$ in \eqref{g0}. The terms containing the integrals are small for this parameter set. 
Moreover, the limits of 
the reflected and transmitted waves $f_1$ and $g_5,$ respectively, as $t\to\infty$ can also be calculated by using \eqref{eq:36_refl} and \eqref{eq:36_tran}.
The time evolution of the solution $x_1(t),x_2(t)$ (i.e., the electron velocities $v_2,v_4$) of the resulting DDE system \eqref{DDEsystem_2delays} and the reflected and transmitted waves are plotted in Figure~\ref{fig:v2v4phi0}, where the initial function $\phi$ is constant zero for both components. 
In the larger plots the long time behavior of the solution is plotted, whilst in the small boxes some snapshots 
are taken at the beginning and at the end of the time interval of the simulation. The oscillations 
appear in wave packages of length $\tau$ and for the velocities $v_2$ and $v_4$ they are shifted with respect to each other with a time interval $\tau$. They indicate that the fragmentation is a result of the accumulation of shifted interferences. 
Their 
amplitude is decreasing in time, following the asymptotic behavior described in \eqref{lim1}, where in this case $c_\phi=0.$      
\begin{figure}
 {\includegraphics[width=.5\columnwidth]{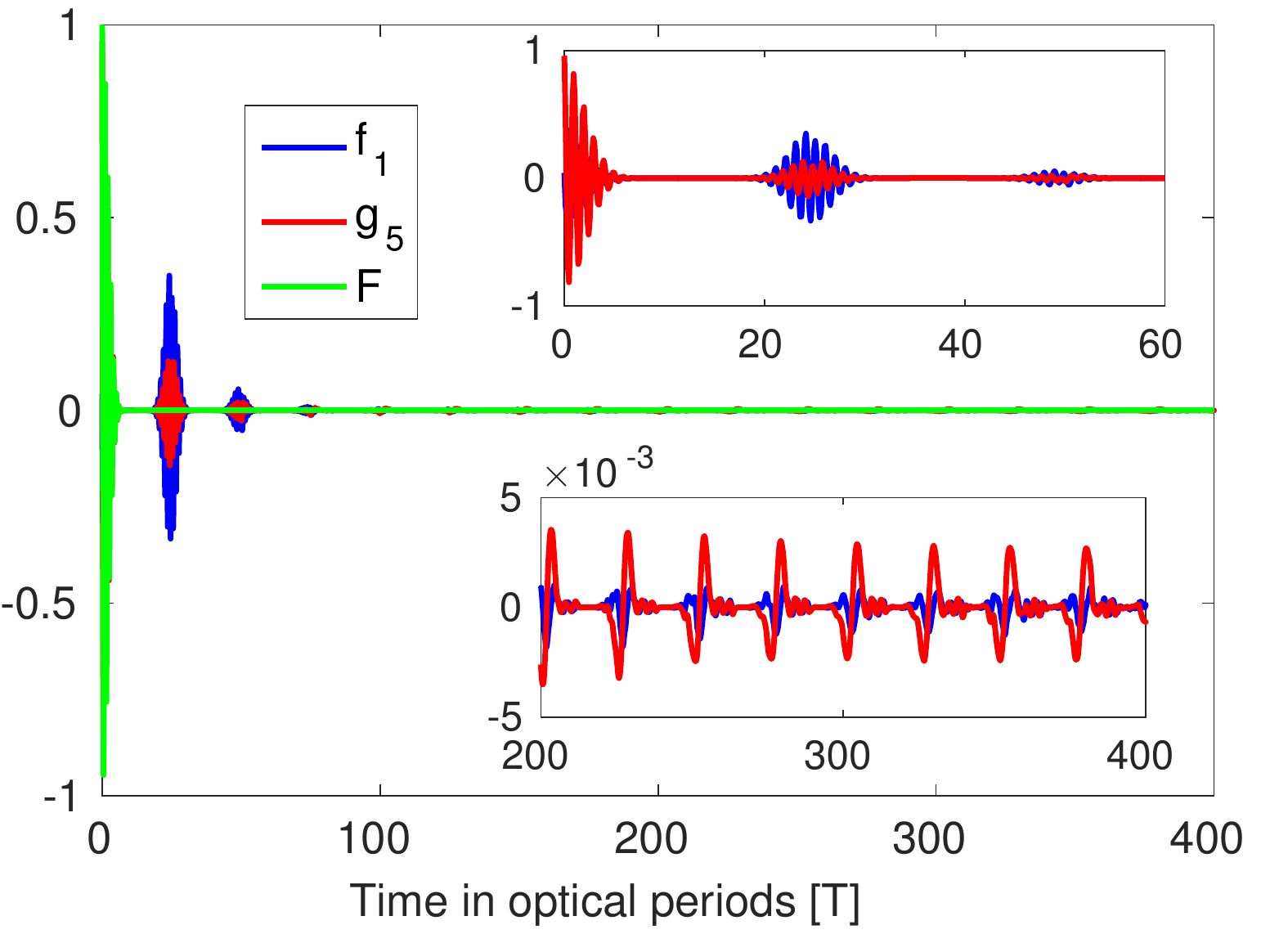}}
{\includegraphics[width=.5\columnwidth]{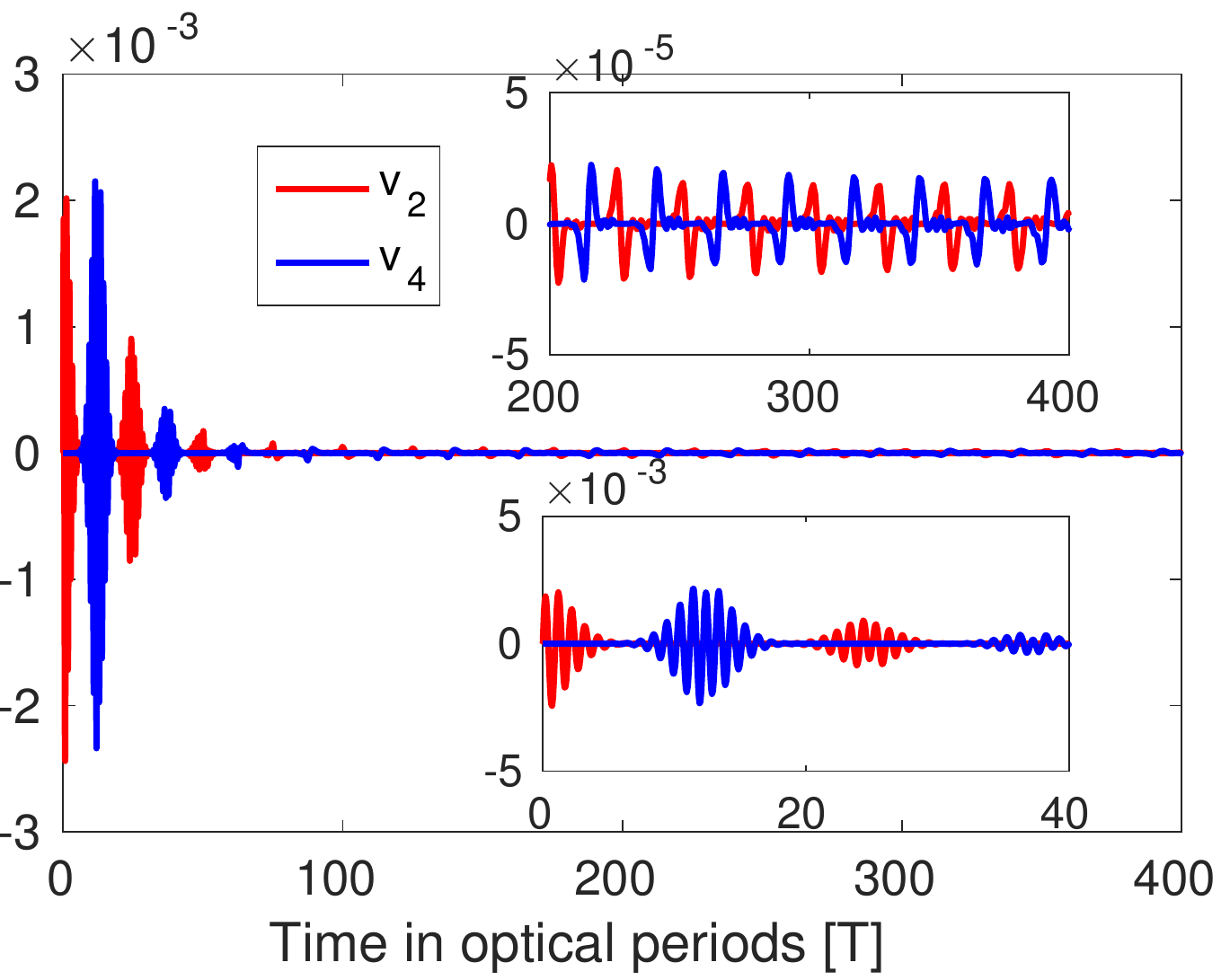}} 
 \caption[]{The solution of the system \eqref{DDEsystem_2delays} (right), the reflected, transmitted and incoming waves (left), with initial function $\phi=(0,0)^T$.   When $n_1=1$,  $n_3=n_5=1.5$,  $\theta_1=\pi/3$ and 
the distance between the layers is $h=10\lambda_0$,  then $\tau=12.25$.  
The layers have thickness $l_i=2$ nm, $i=2,4$. }\label{fig:v2v4phi0}
\end{figure}
Figure~\ref{fig:v2v4phi_id} illustrates the solution of the system when the initial function is the eigenvector 
$\phi=(1,-r_2/r_4)^T$. In this case $c_\phi=1$, and by choosing layers with the same thickness, i.e., $r_2=r_4,$ the limiting behavior \eqref{lim1} can be observed.  
\begin{figure}
  \centering
{\includegraphics[width=.5\columnwidth]{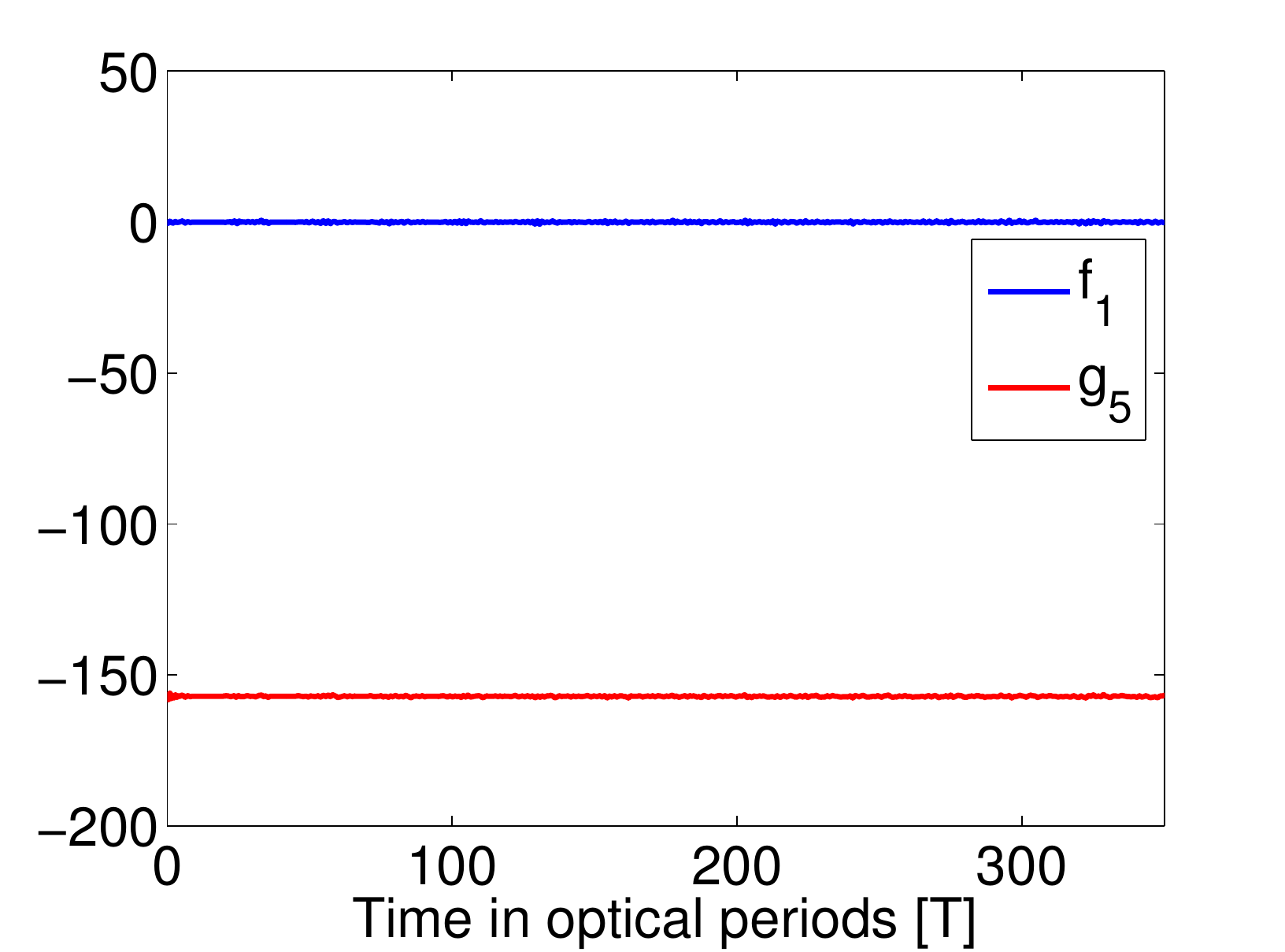}}
{\includegraphics[width=.5\columnwidth]{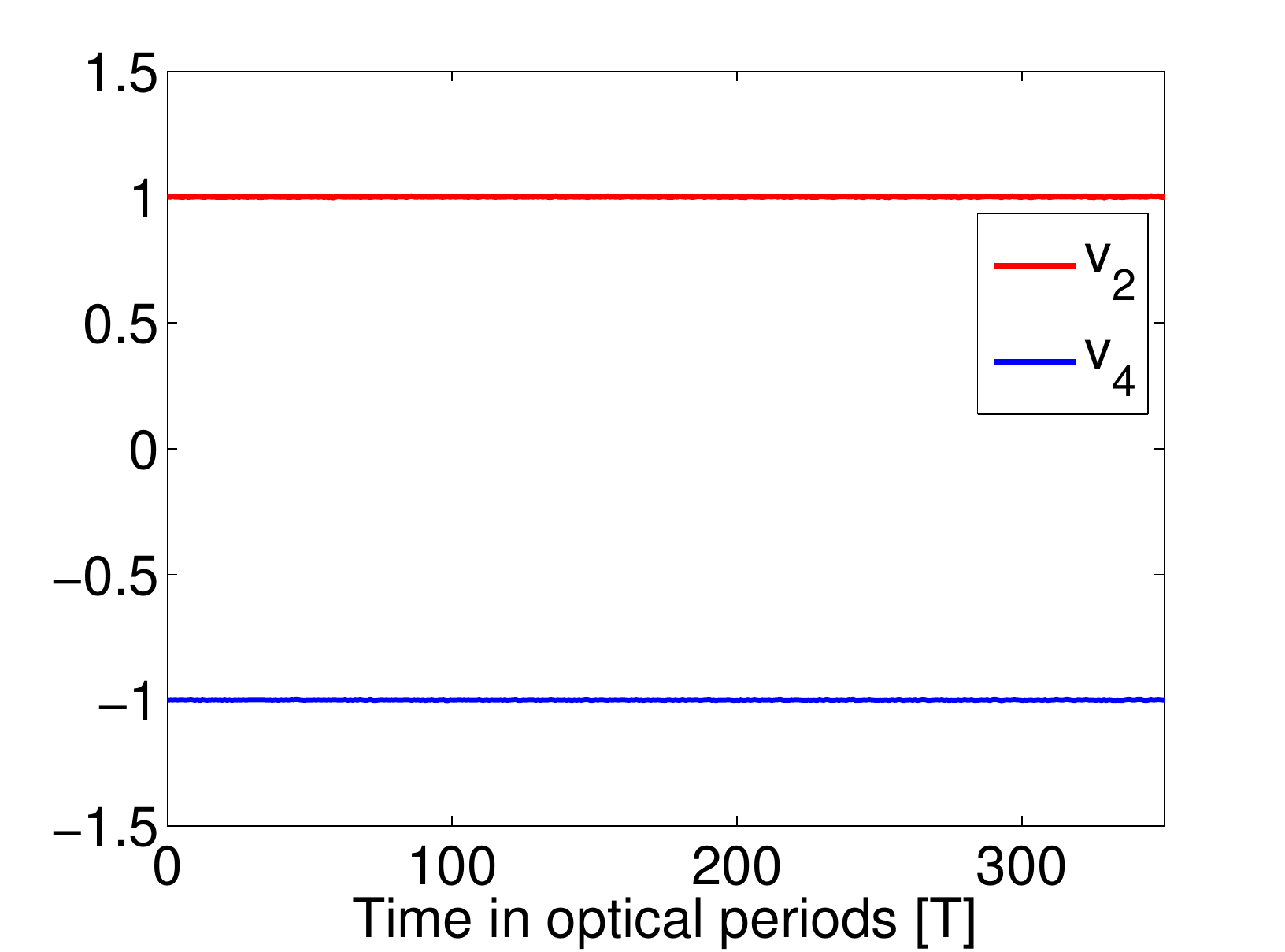}} 
  \caption[]{The solution of the system \eqref{DDEsystem_2delays} (right), the reflected and transmitted waves (left), with initial function $\phi=(1,-1)^T$.  When $n_1=1$,  $n_3=n_5=1.5$,  $\theta_1=\pi/3$ and 
the distance between the layers is $h=10\lambda_0$,  then $\tau=12.25$.  The layers have thickness $l_i=2$ nm,  
$i=2,4$. }\label{fig:v2v4phi_id}
\end{figure} 
In Figure~\ref{fig:v2v4phi_const} the solution of the system is plotted when  starting with the constant initial function $\phi=(0.02,0.01)^T$. The snapshots demonstrate the oscillatory behavior on short time intervals.
\begin{figure}
  \centering
{\includegraphics[width=.49\columnwidth]{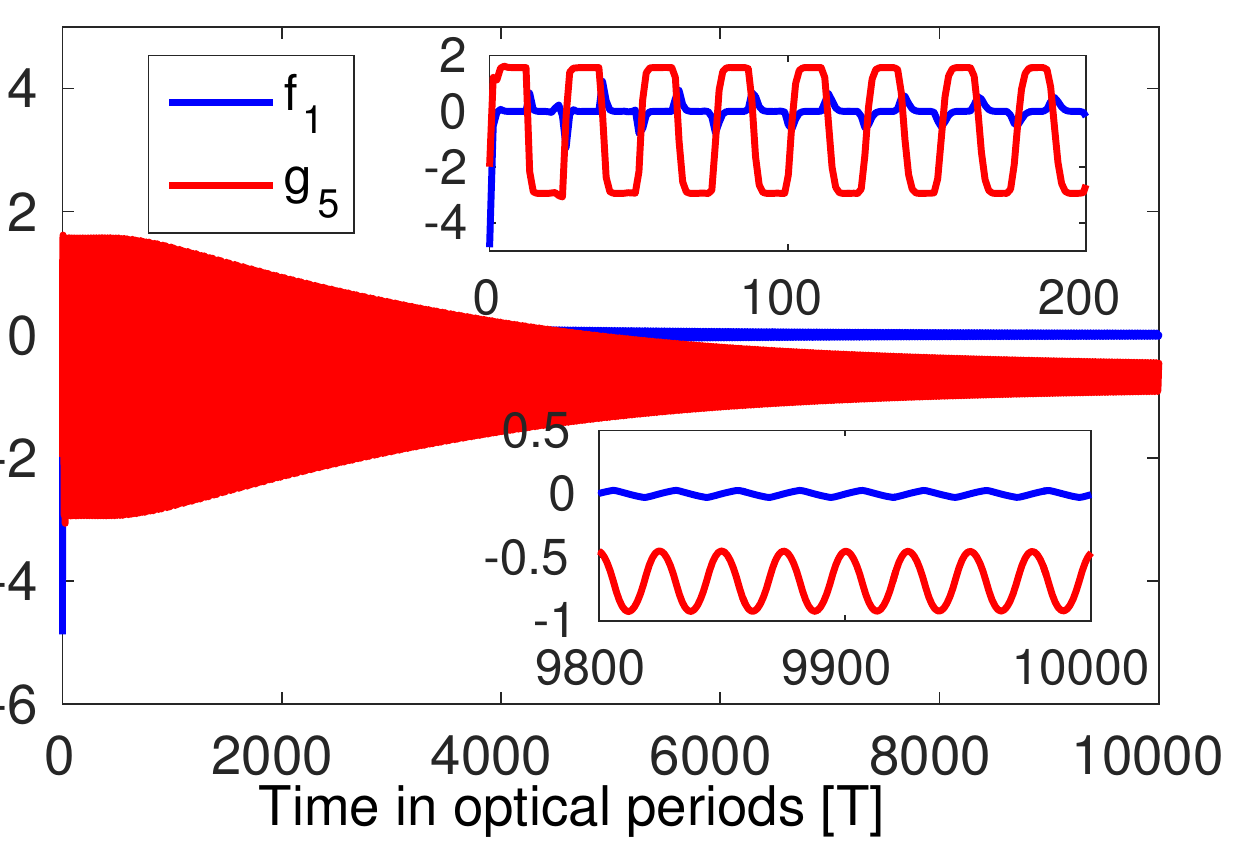}}
{\includegraphics[width=.49\columnwidth]{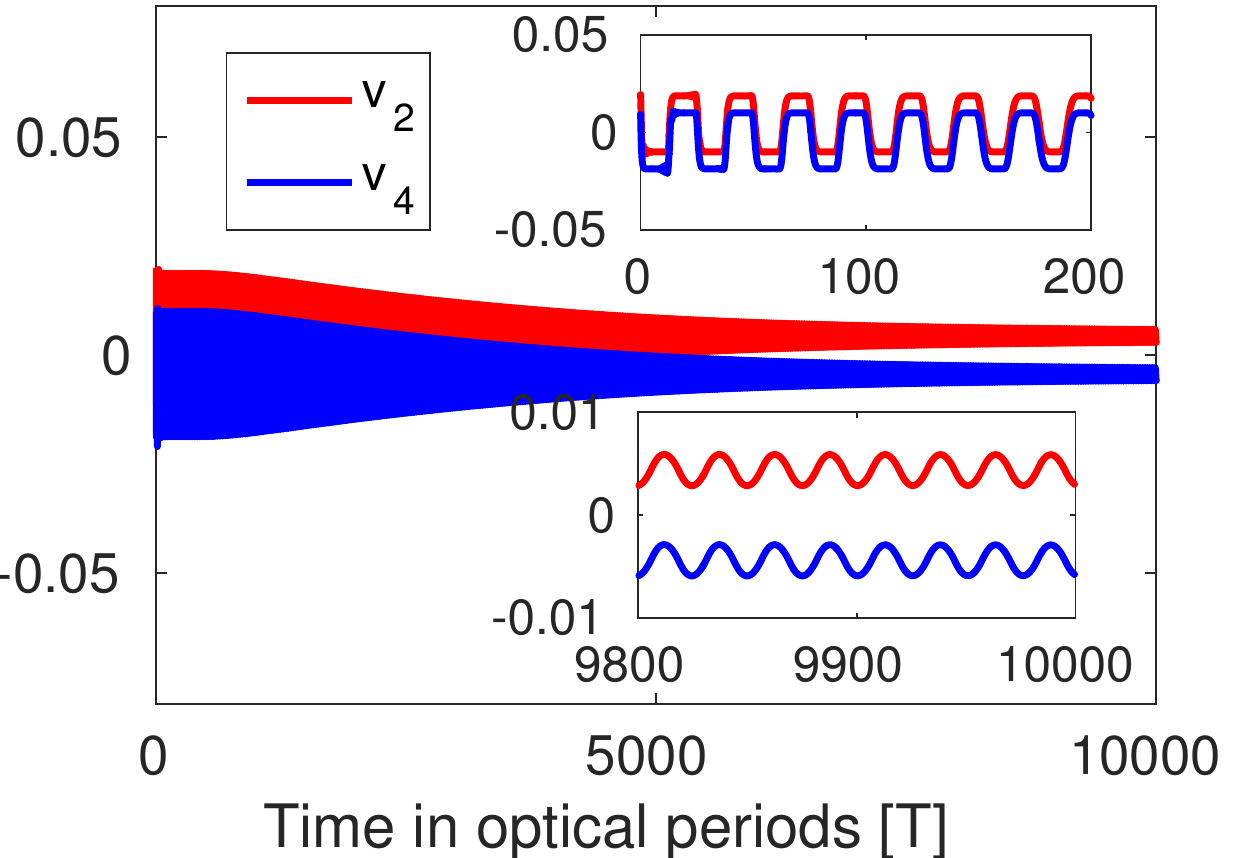}}
  \caption[]{The solution of the system \eqref{DDEsystem_2delays} (right), the reflected and transmitted waves (left), with initial function $\phi=(0.02,0.01)^T$. When $n_1=1$, $n_3=n_5=1.5$,  $\theta_1=\pi/3$ and 
the distance between the layers is $h=10\lambda_0$, then $\tau=12.25$. The layers have thickness $l_i=2$ nm,  
$i=2,4$. }\label{fig:v2v4phi_const}
\end{figure} 

The thickness of the second metal layer is then reduced $l_4=1$ nm whilst the thickness of $l_2$ unchanged. The initial function is constant $\phi=(0.02,0.01)^T$ and the solution of the system is plotted in Figure~\ref{fig:v2v4phi_const_r4}. Compared to the 
previous case, we conclude that the layer thickness greatly influences the limiting behavior.  
\begin{figure}
  \centering
{\includegraphics[width=.495\columnwidth]{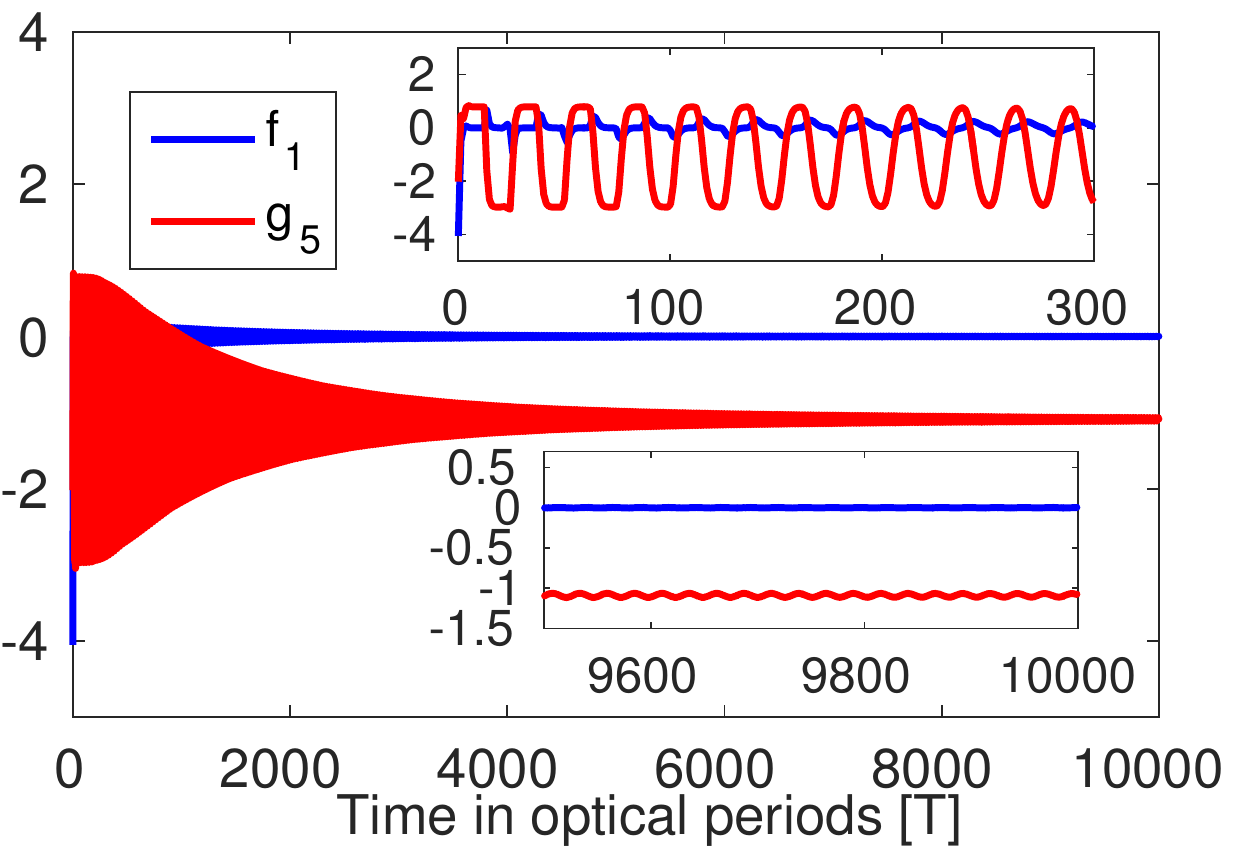}}
{\includegraphics[width=.495\columnwidth]{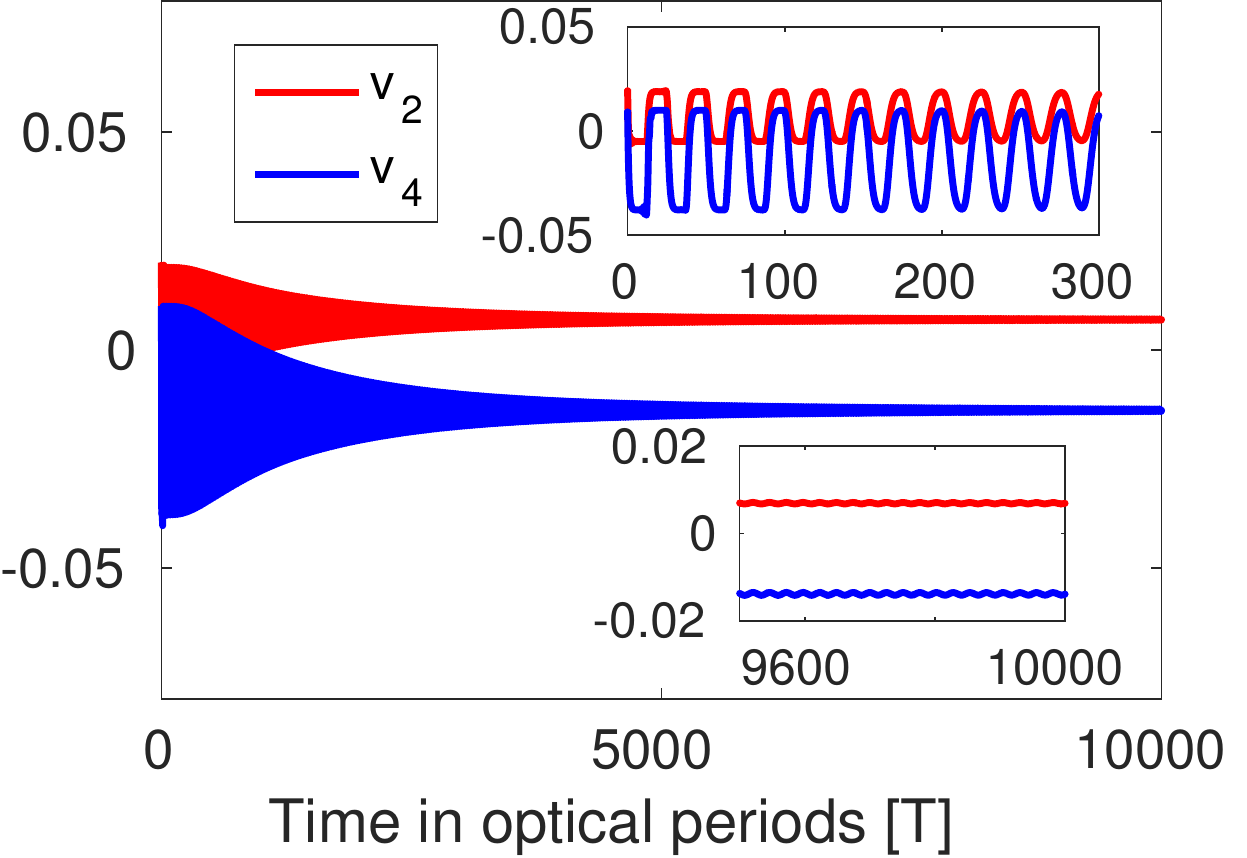}}
\caption[]{The solution of the system \eqref{DDEsystem_2delays}
  (right), the reflected and transmitted waves (left), with initial
  function $\phi=(0.02,0.01)^T$.  When $n_1=1$,  $n_3=n_5=1.5$,  $\theta_1=\pi/3$ and 
the distance between the layers is $h=10\lambda_0$,  then $\tau=12.25$.  The layers have thickness $l_2=2$ nm  and $l_4=1$ nm,  respectively.}\label{fig:v2v4phi_const_r4}
\end{figure}

Finally, Figure~\ref{fig:v2v4phi_time} shows the solution when starting with non-constant initial electron velocities,
$\phi=(1+0.1t, 0.01t)^T$. This is an illustrative example where \eqref{solution_limit} needs to be used to compute the limit.
\begin{figure}
  \centering
{\includegraphics[width=.49\columnwidth]{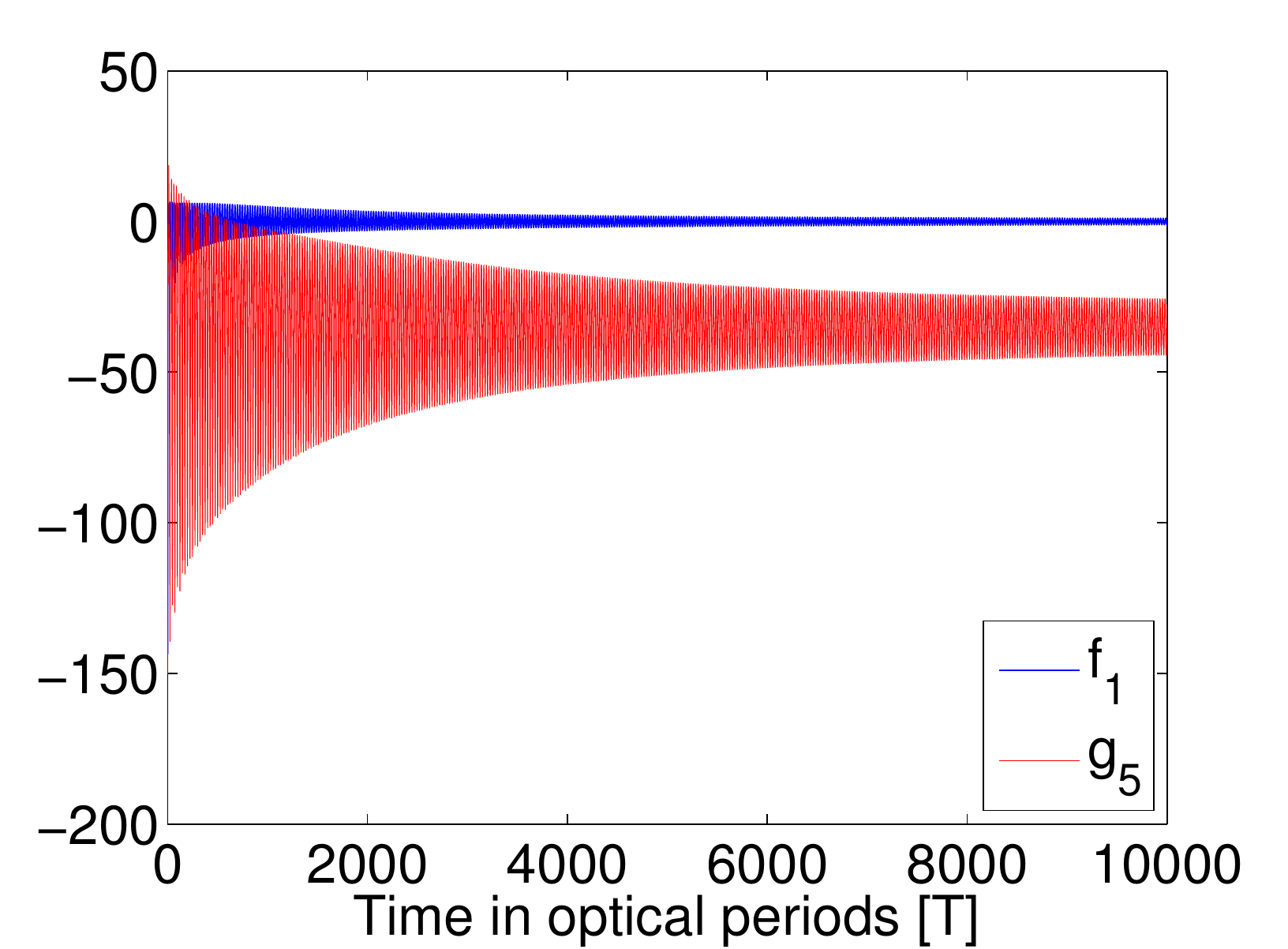}}
{\includegraphics[width=.49\columnwidth]{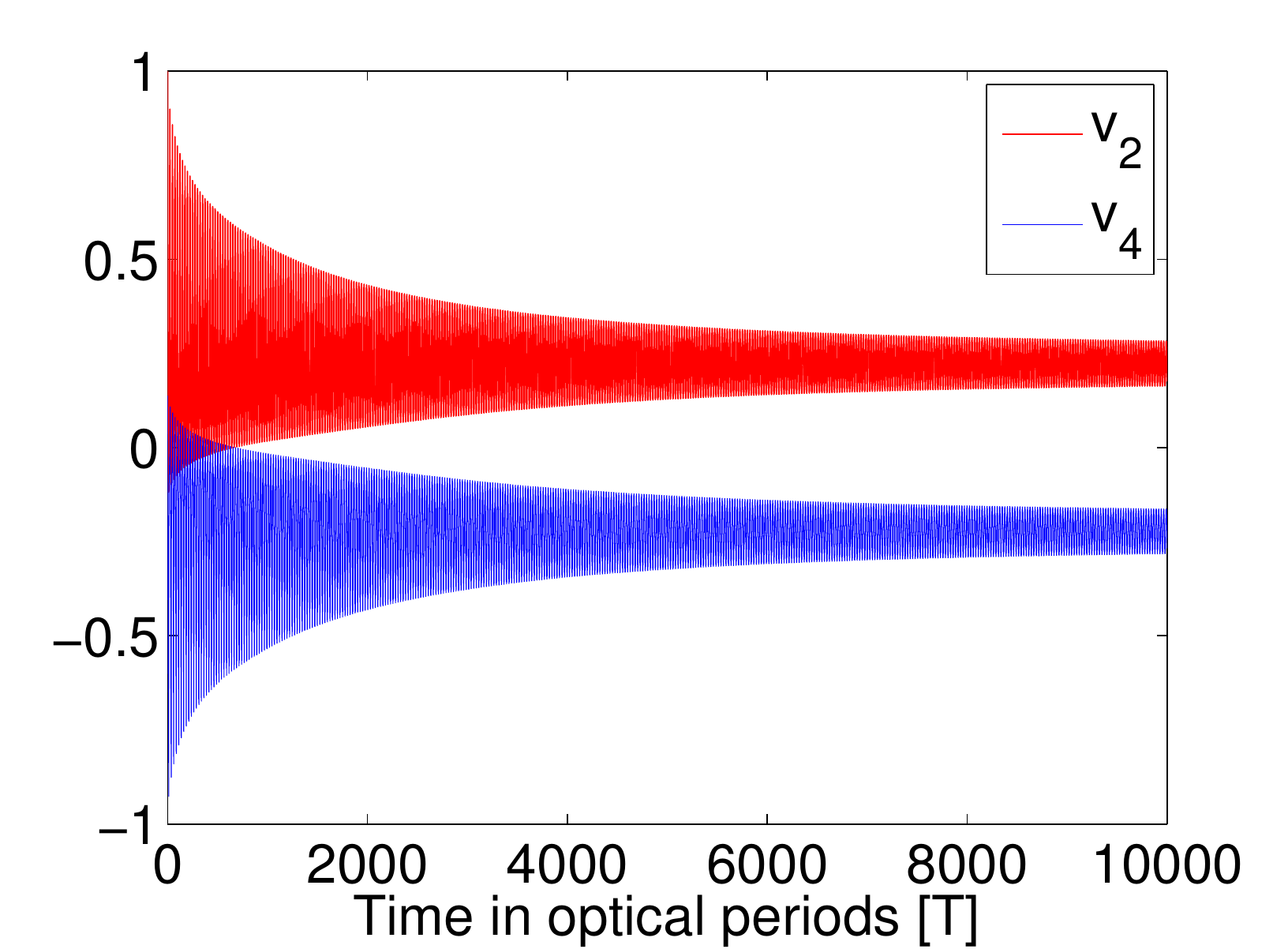}} 
  \caption[]{The solution of the system \eqref{DDEsystem_2delays} (right), the reflected and transmitted waves (left), with initial function $\phi=(1+0.1t, 0.01t)^T$.  When $n_1=1$, $n_3=n_5=1.5$,  $\theta_1=\pi/3$ and 
the distance between the layers is $h=10\lambda_0$,  then $\tau=12.25$.  
The layers have thickness $l_i=2$ nm,  $i=2,4$. }\label{fig:v2v4phi_time}
\end{figure}

\noindent {\em Case 2.} This case considers the scenario when $n_1=1$, $n_3=1.1$ and $n_5=1.5$. The thicknesses of the layers are equal, and the HS system \eqref{DDDE} is solved numerically for three cases: the zero initial function in Figure~\ref{fig:gen_v2v4phi_const0}, 
for the non-zero initial function $\phi=(0.02,0.01,0.03)^T$, Figure~\ref{fig:gen_cons_v2v4}, and 
when it is 
the eigenvector $\phi=(1,-r_2/r_4, r_2/a_0)^T$ corresponding to the zero eigenvalue, Figure~\ref{fig:gen_v2v4phi_eig}.    
\begin{figure}
  \centering
{\includegraphics[width=.5\columnwidth]{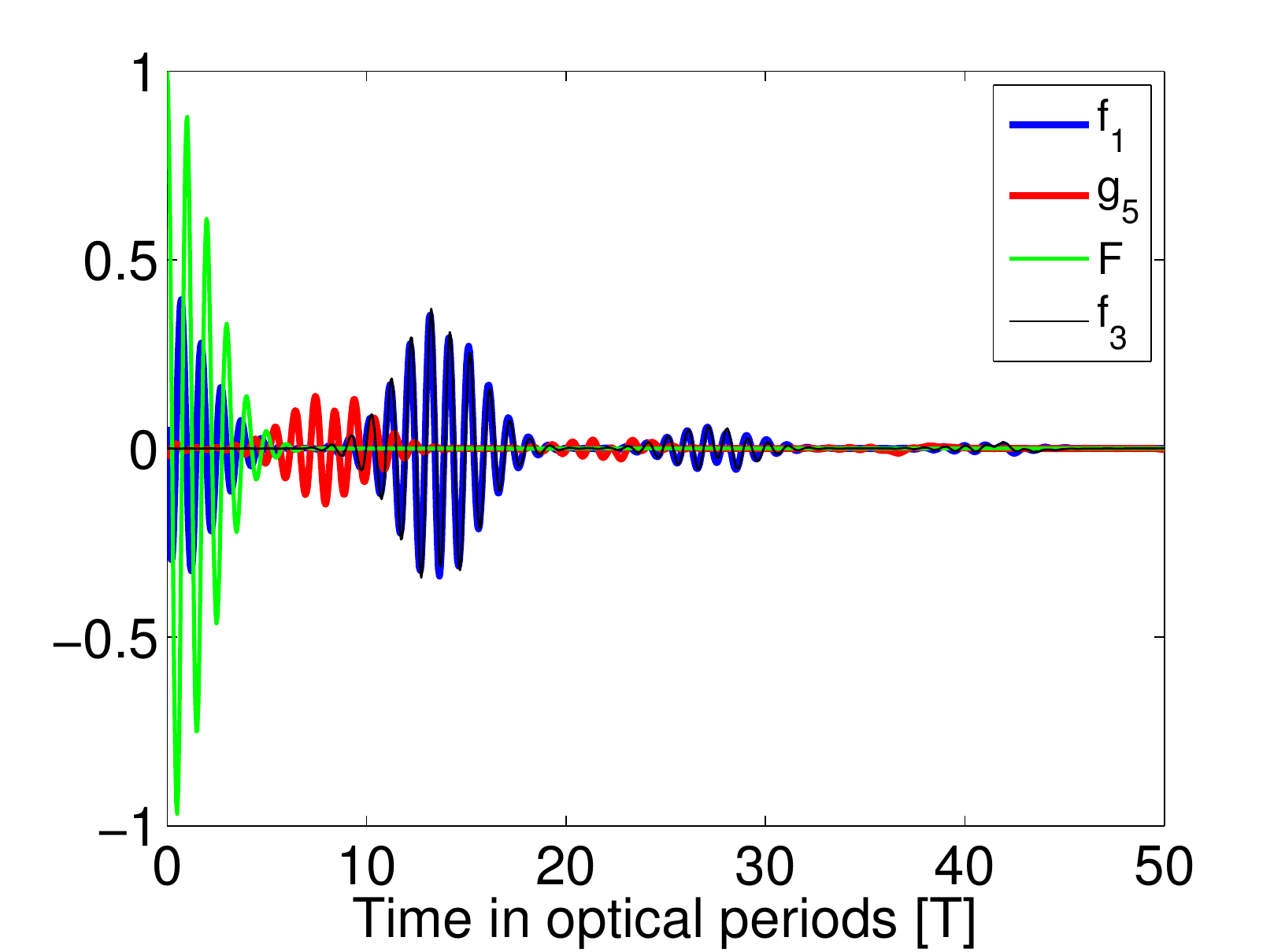}} 
{\includegraphics[width=.49\columnwidth]{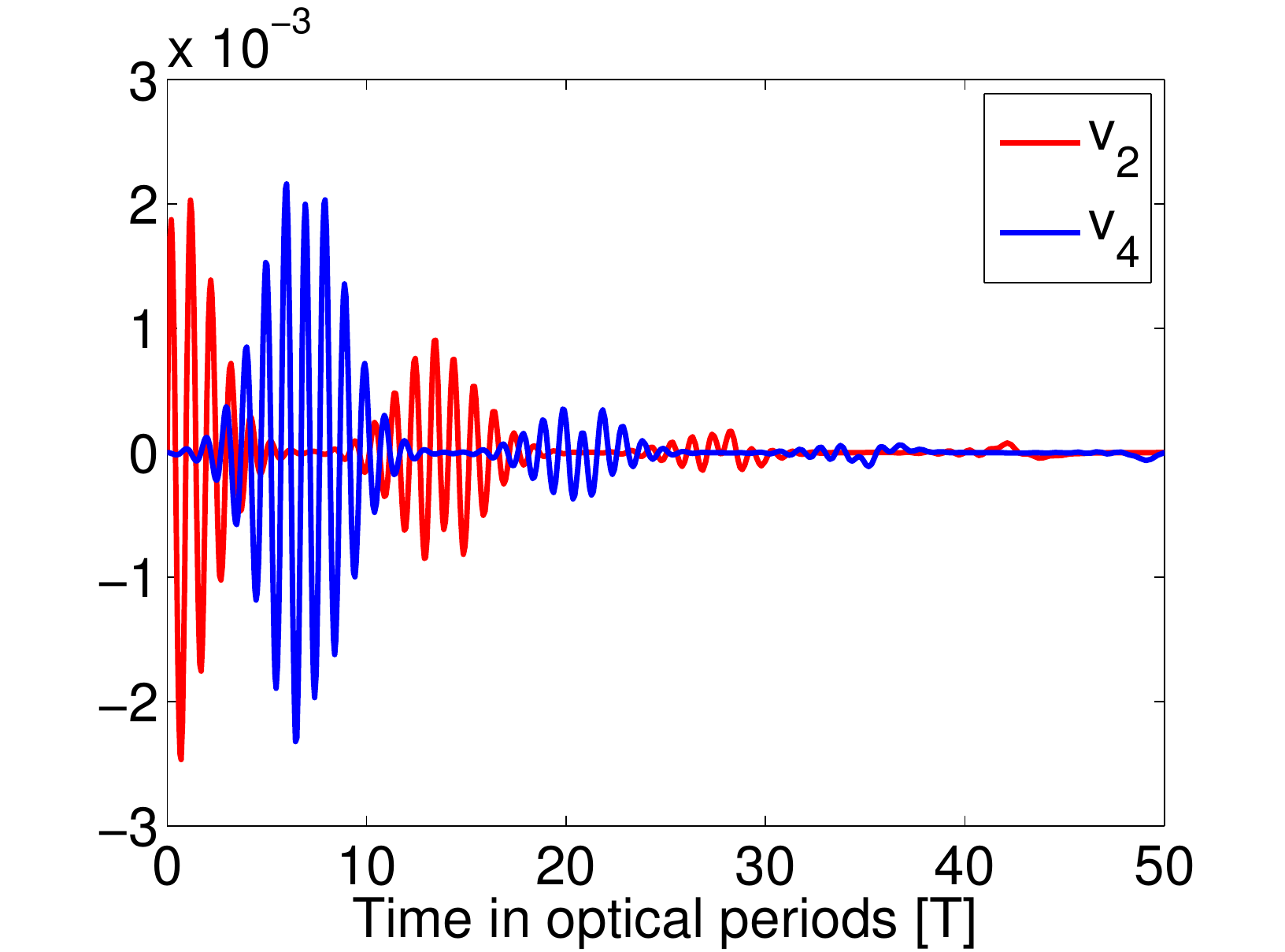}} 
\caption{The solution of the system \eqref{DDDE}
  (right), the reflected and transmitted waves (left), with initial
  function $\phi=(0,0,0)^T$.  When $n_1=1$ $n_3=1.1$,  $n_5=1.5$,  $\theta_1=\pi/3$ and 
the distance between the layers is $h=10\lambda_0$,  then $\tau=6.78$.  
The layers have thickness $l_i=2$ nm,  $i=2,4$. }\label{fig:gen_v2v4phi_const0}
\end{figure}
\begin{figure}
  \centering
{\includegraphics[width=.495\columnwidth]{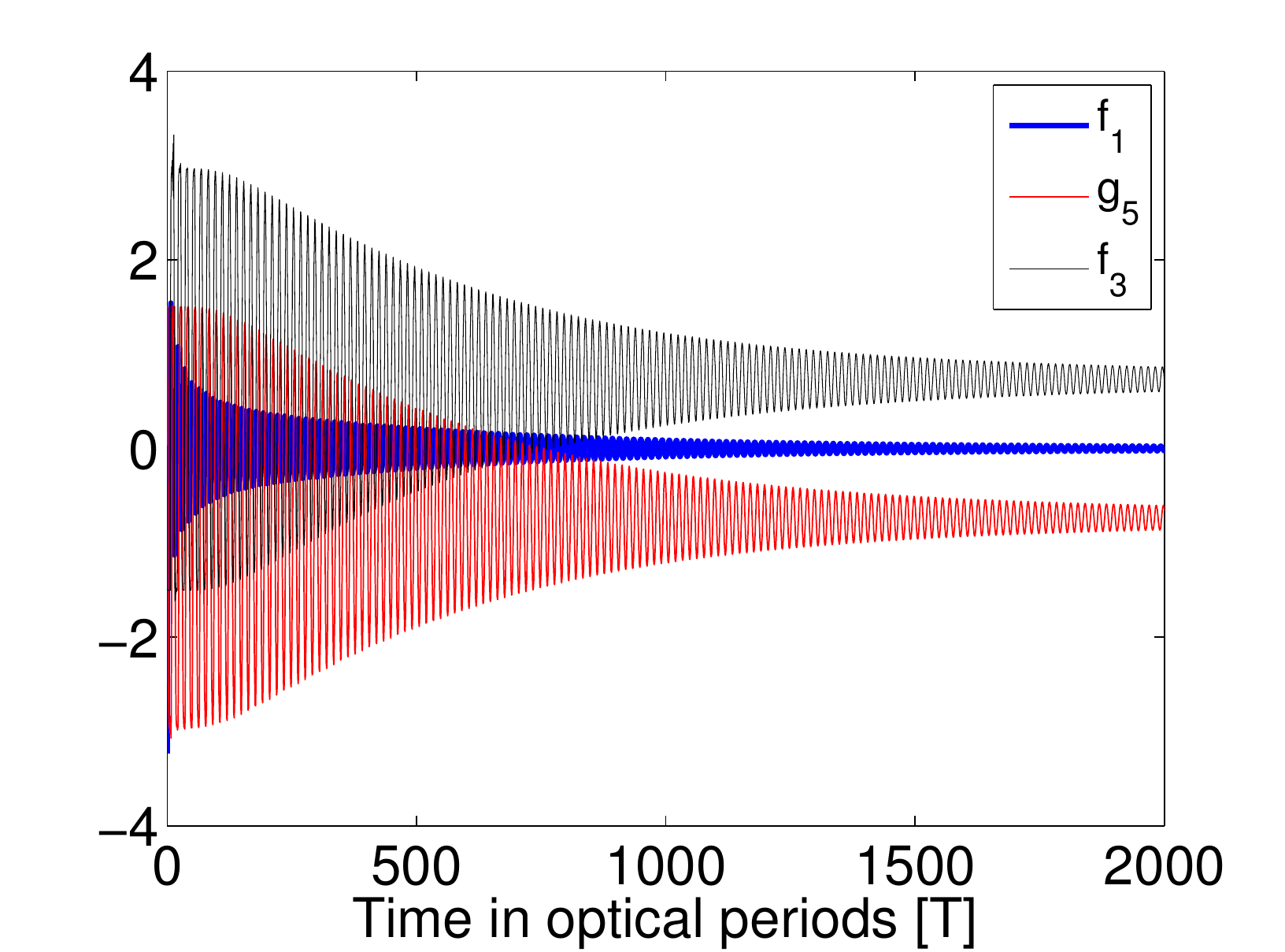}} 
{\includegraphics[width=.495\columnwidth]{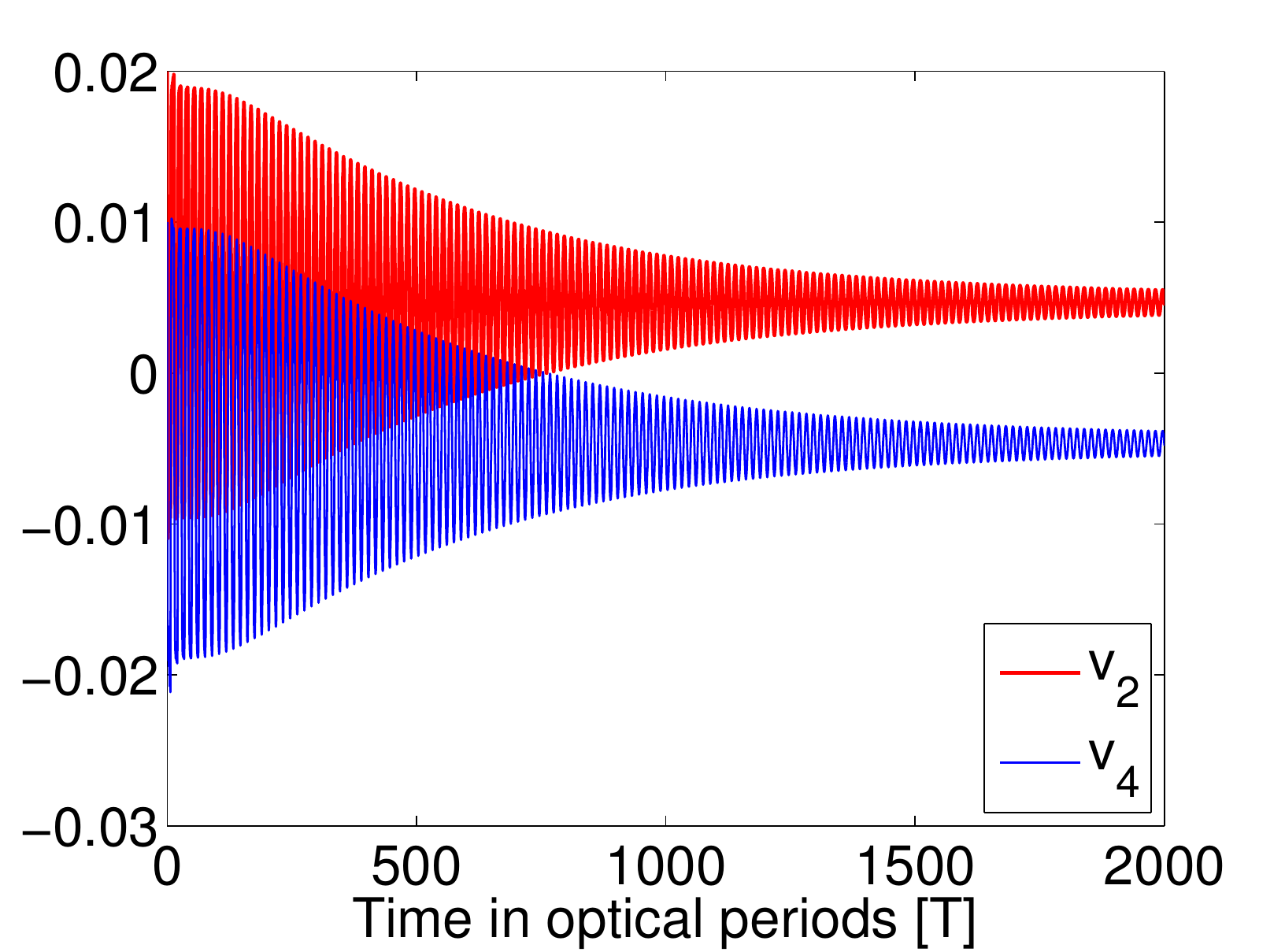}} 
\caption{The solution of the system \eqref{DDDE}
  (right), the reflected and transmitted waves (left), with initial
  function $\phi=(0.02,0.02,0.03)^T$.  When $n_1=1$ $n_3=1.1$,  $n_5=1.5$,  $\theta_1=\pi/3$ and 
the distance between the layers is $h=10\lambda_0$,  then $\tau=6.78$.  
The layers have thickness $l_i=2$ nm,  $i=2,4$. }\label{fig:gen_cons_v2v4}
\end{figure}
\begin{figure}
  \centering
{\includegraphics[width=.49\columnwidth]{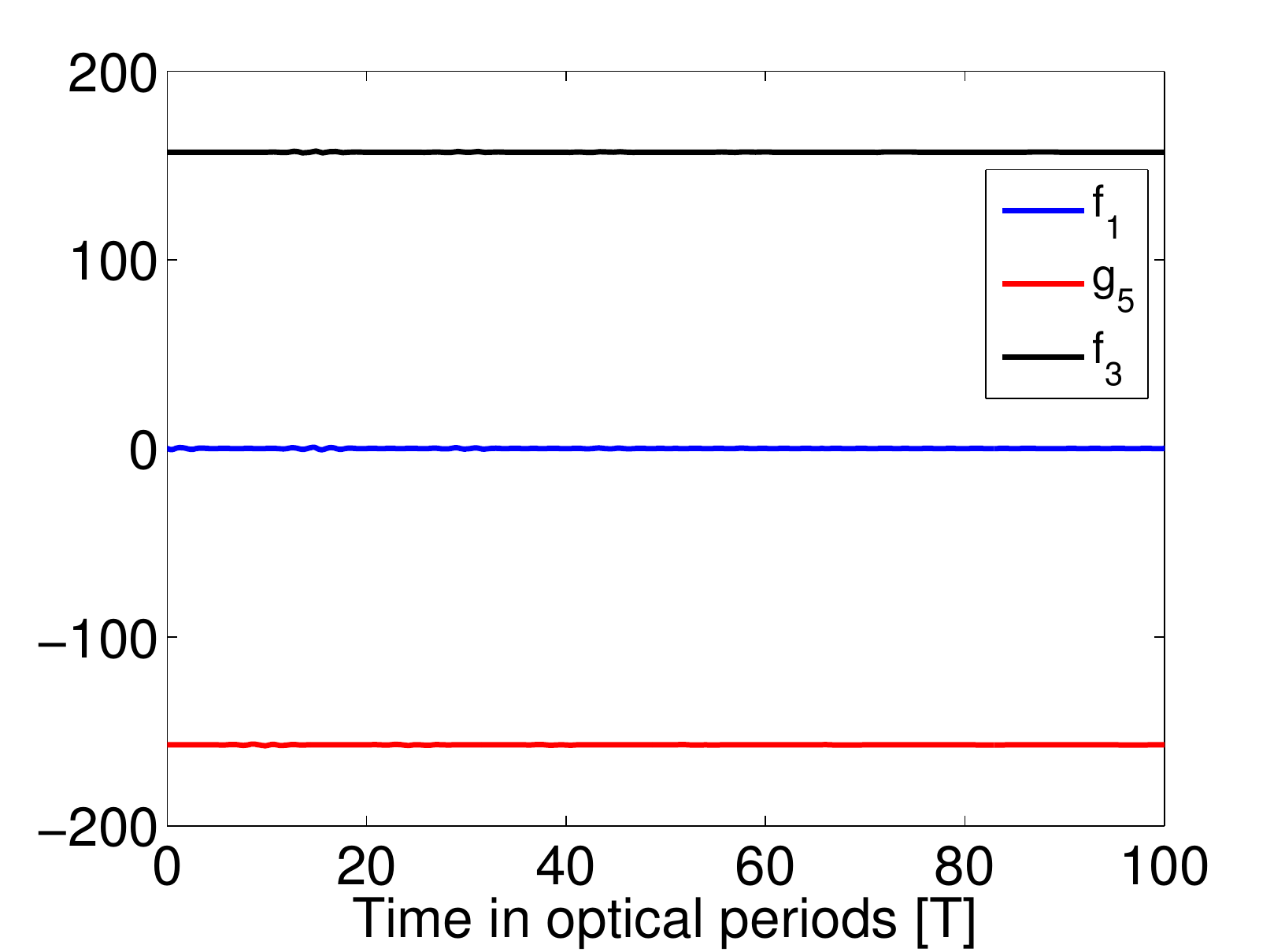}} 
{\includegraphics[width=.49\columnwidth]{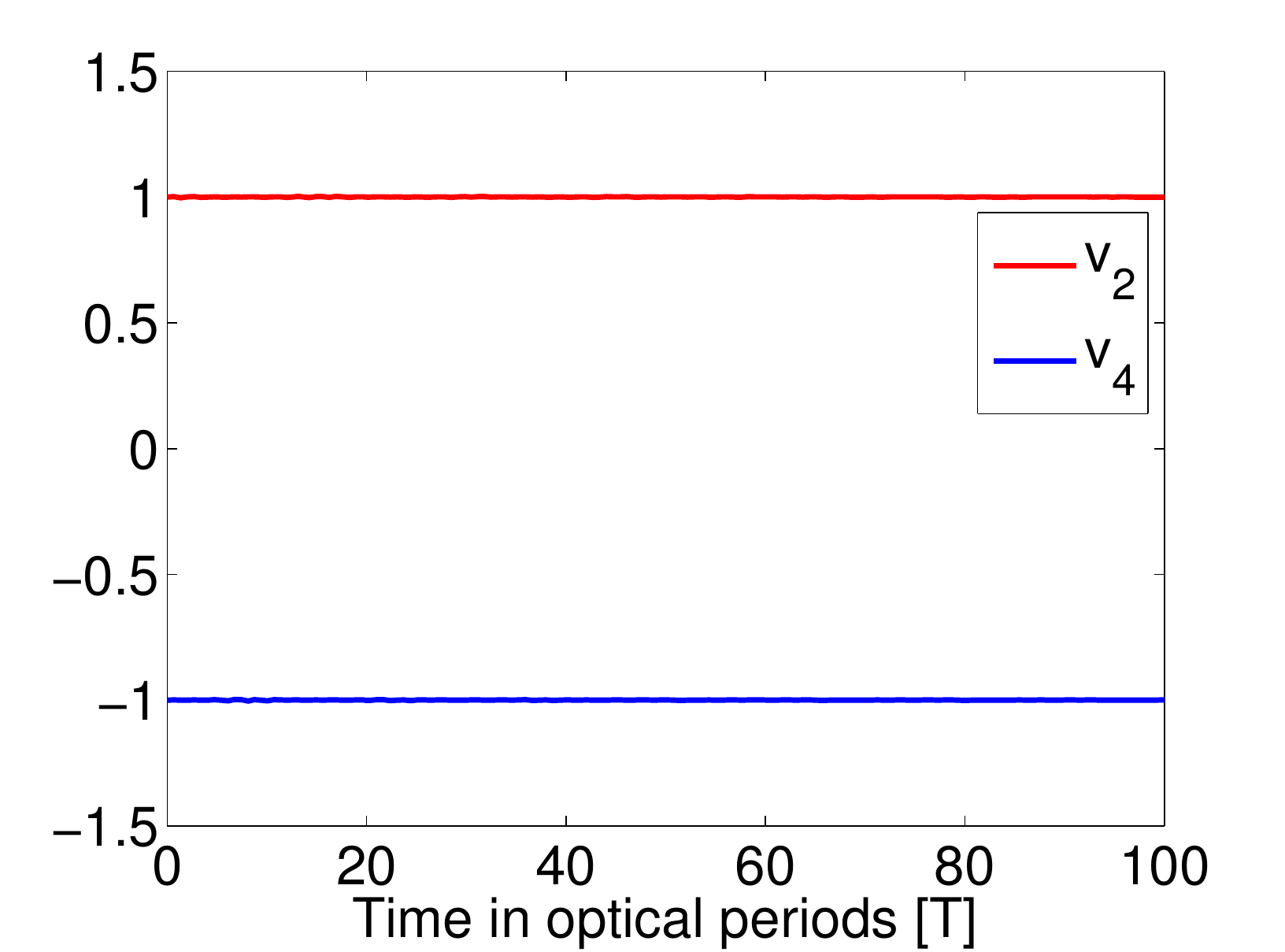}} 
\caption{The solution of the system \eqref{DDDE}
  (right), the reflected and transmitted waves (left), with initial
  function $\phi=(1,-r_2/r_4, r_2/a_0)^T$,  for $n_1=1$ $n_3=1.1$,  $n_5=1.5$. When $n_1=1$ $n_3=1.1$, $n_5=1.5$,  $\theta_1=\pi/3$ and 
the distance between the layers is $h=10\lambda_0$,  then $\tau=6.78$.  
The layers have thickness $l_i=2 nm$, $i=2,4$.}\label{fig:gen_v2v4phi_eig}
\end{figure}
The limiting behaviors were verified in all examples by setting $\epsilon=0$ in Remark~\ref{r:Pphi}.

\noindent {\em Case 3.} The aim of this last set of examples is to briefly analyze the solution in the 
frequency domain. Let $n_1=n_5,$ the incidence angle $\theta_1=\pi/6,$ and the initial function $\phi=(0,0)^T$. 
Here the dependence on the CE phase difference of the reflected flux is observed at relative minimum points of the spectrum. The representation of the solution $x$ of the DDE system \eqref{DDEsystem_2delays} in the frequency domain is obtained by setting the real part of the complex variable $s$ in the Laplace transform to zero, i.e., $X(i\omega)= \Delta^{-1}(i\omega)H(i\omega)$, $\omega\not=0$. By using the information 
on the roots of the characteristic equation, this transform exists. Taking the Laplace transform of the reflected and transmitted waves, we 
obtain that in the frequency domain
\begin{align}
{F_1}(i\omega)&:= \mathcal{L}(f_1)(i\omega)=\frac{1}{c_1+c_3}
\left[(c_3-c_1)\mathcal{L} F(i\omega)
-2 c_3
\frac{r_2}{a_0} X_1 (i\omega)
-2 c_3 \frac{r_4}{a_0}e^{-i\omega \tau} X_2(i\omega)
\right],\nonumber\\
{G_5}(i\omega)&:=\mathcal{L}(g_5)(i\omega)=\frac{2 c_1}{c_1+c_3}
\left[\mathcal{L} F(i\omega)
-
\frac{r_2}{a_0} X_1(i\omega)
\right]
-\frac{c_1-c_3}{c_1+c_3} \frac{r_4}{a_0}e^{-i\omega \tau} X_2(i\omega),
\end{align}
where $X_1$ and $X_2$ are the components of $X.$

In the first example, all indices of refraction are equal and set to 1.5, and the thickness of the layers are $l_2=2$ nm and $l_4=1$ nm, respectively. The dependence of the reflected, transmitted and incoming flux on the normalized frequency $v=\omega/\omega_0$ is plotted in Figure~\ref{fig:transforms} (left). Considerable modulation always appears at the local minima of the spectrum. In Figure~\ref{fig:CEP_equal}, the modulation of the reflected flux can be observed at the minimum $v=0.98$ and $v=1.06,$ respectively, as a function of the CE phase. The modulation function is (see \cite{Varro2007carrier})
\[
M(\omega)=\frac{I_{max}(\omega)-I_{min}(\omega)}{I_{max}(\omega)+I_{min}(\omega)},
\]
where $I_{max},I_{min}$ are the maximum and minimum of the reflected flux as functions of the 
CE phase, at a given frequency $\omega.$ 
\begin{figure}
 \centering
{\includegraphics[width=.495\columnwidth]{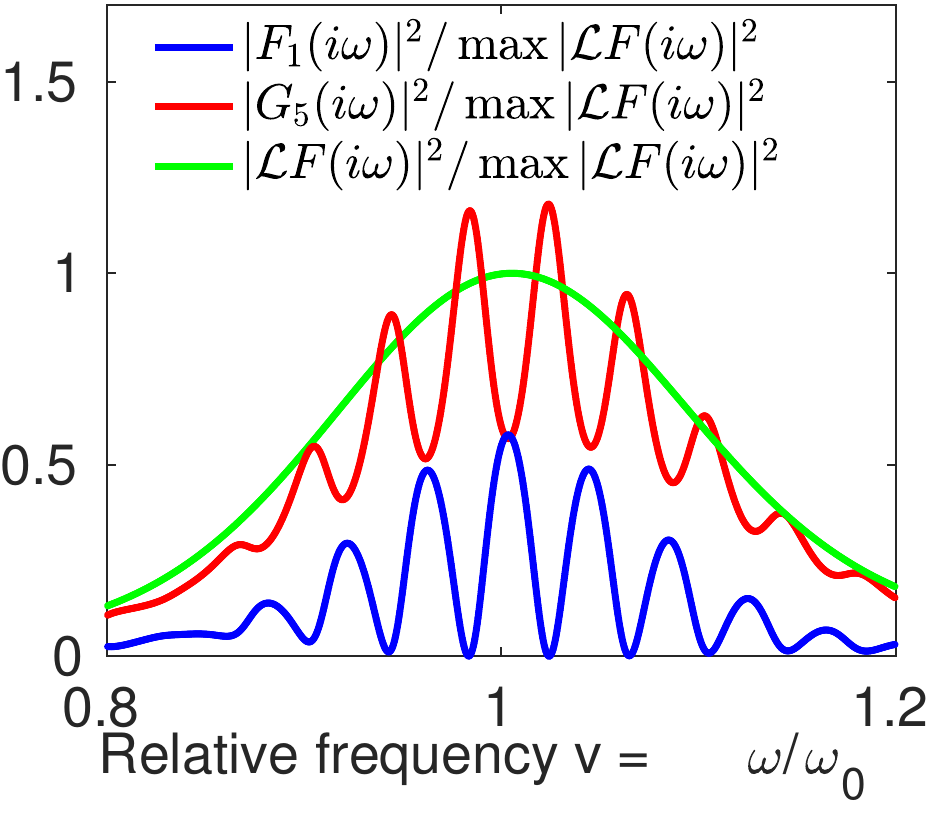}}
{\includegraphics[width=.49\columnwidth]{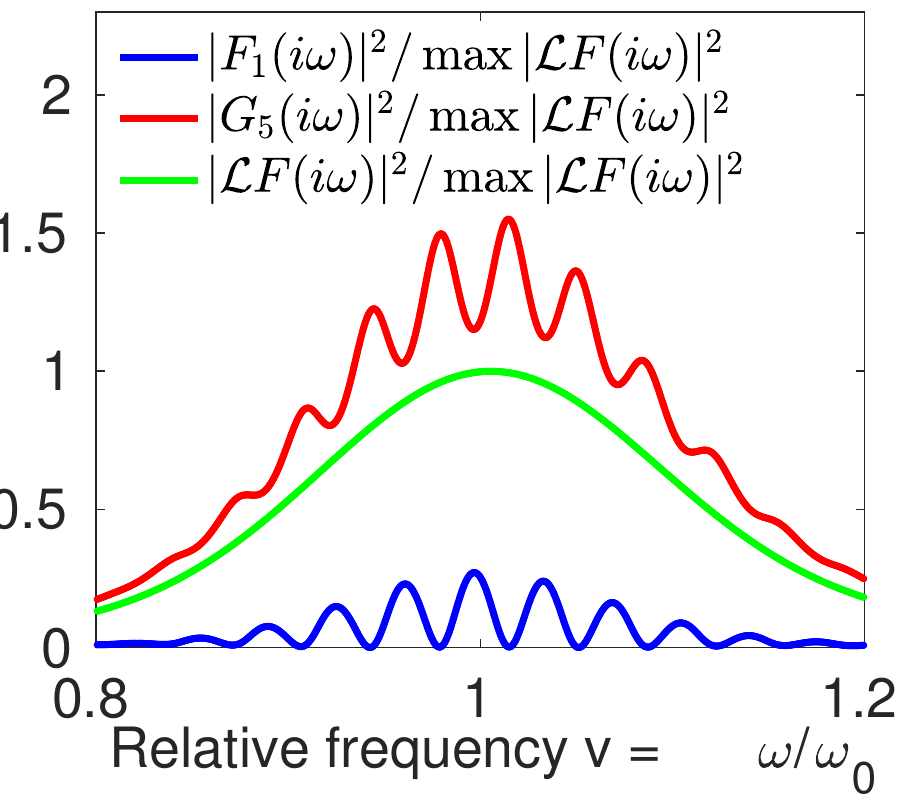}}
\caption{The spectrum of the reflected, transmitted and incoming pulses, with $n_1=n_3=n_5=1.5$, 
$\theta_1=\pi/6,$ $l_2=2$ nm and $l_4=1$ nm (left) and $n_1=1,$ $n_3=n_5=1.5$, $\theta_1=\pi/8,$ $l_2=l_4=0.6$ nm (right).}\label{fig:transforms}
\end{figure}
\begin{figure}
 \centering
{\includegraphics[width=.49\columnwidth]{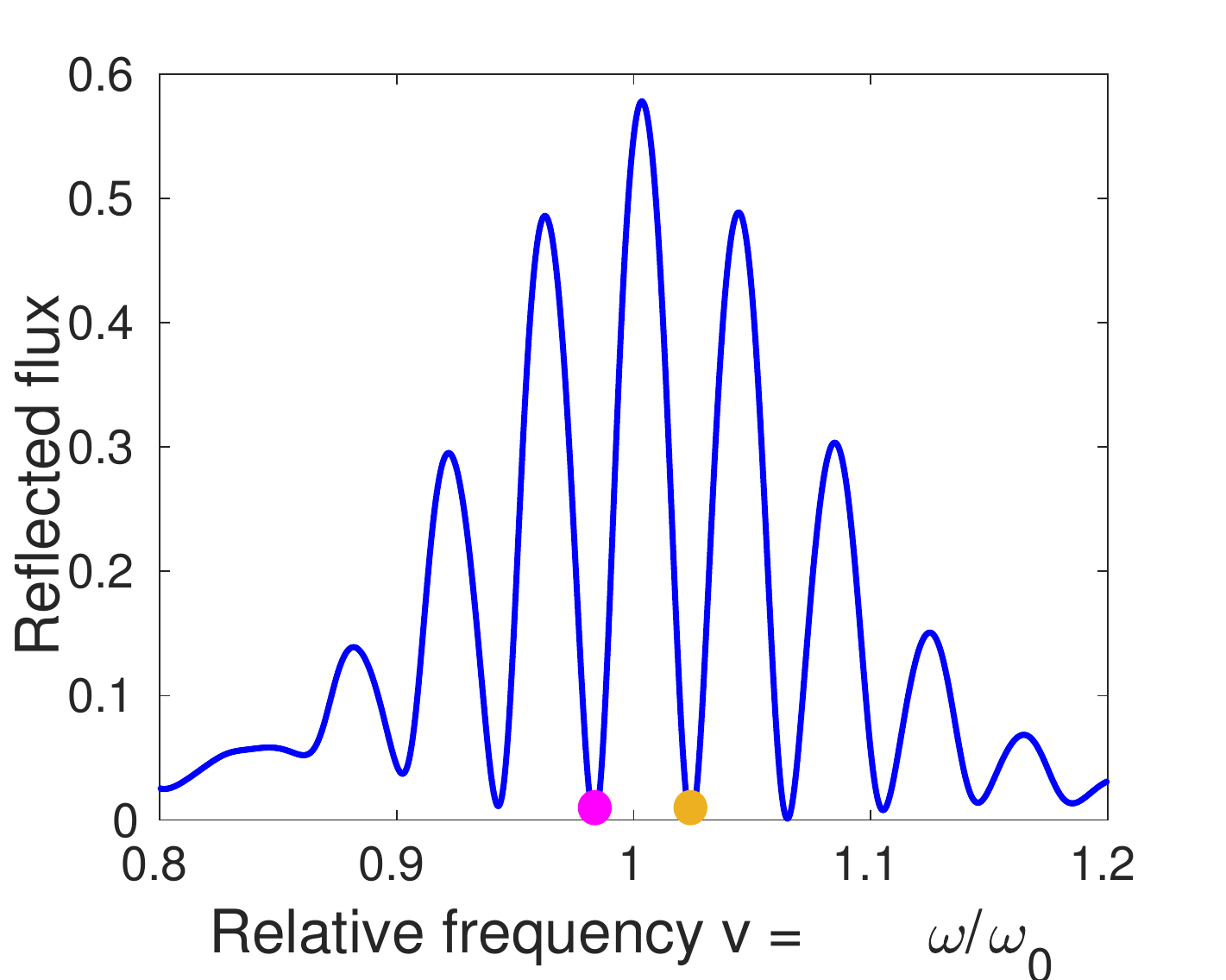}} 
{\includegraphics[width=.49\columnwidth]{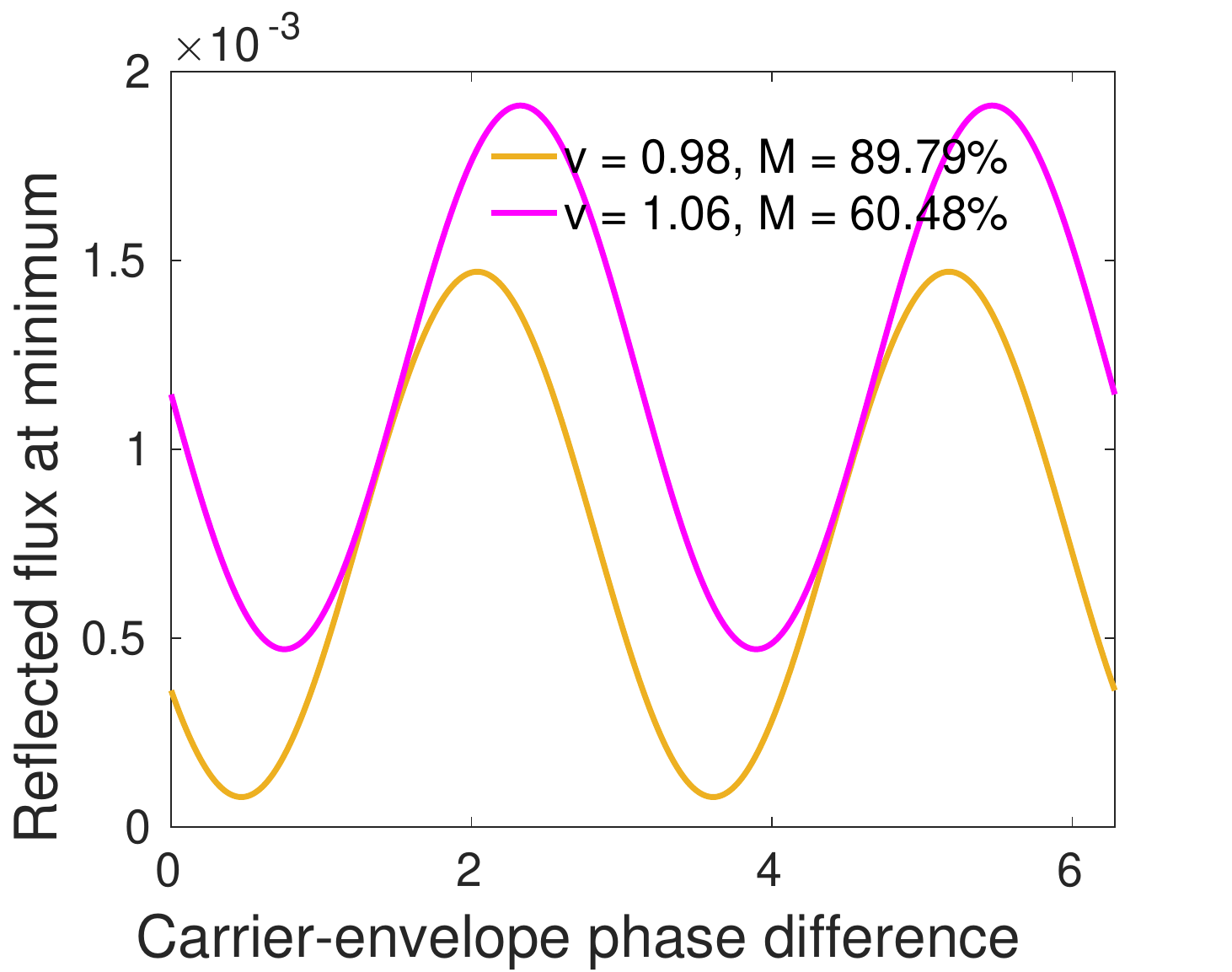}} 
\caption{Dependence of the reflected flux on the CE phase difference at the minimum $v=0.98$ and $v=1.06,$ respectively. The refraction indices are $n_1=n_3=n_5=1.5,$ $\theta_1=\pi/6$ and 
the thickness of the layers is $l_2=2$ nm and $l_4=1$ nm, respectively.}\label{fig:CEP_equal}
\end{figure}

In the second example, we set the indices of refraction to $n_1=1,$
$n_3=n_5=1.5$ and the thickness of the layers to $l_2=l_4=0.6$ nm. The
spectrum is plotted in Figure~\ref{fig:transforms} (right) and the CE
phase dependence of the modulation at the indicated minimum frequencies in Figure~\ref{fig:CEP}. We can observe that the modulation of the side bands is the highest at $v=0.98$ with a value $M=89.79\%$ in the first case and at $v=1.05$ with a value $M=87.81\%$ in the second example. This means that such modulations should be measurable effects for a few-cycle laser pulse.
 \begin{figure}
 \centering
{\includegraphics[width=.49\columnwidth]{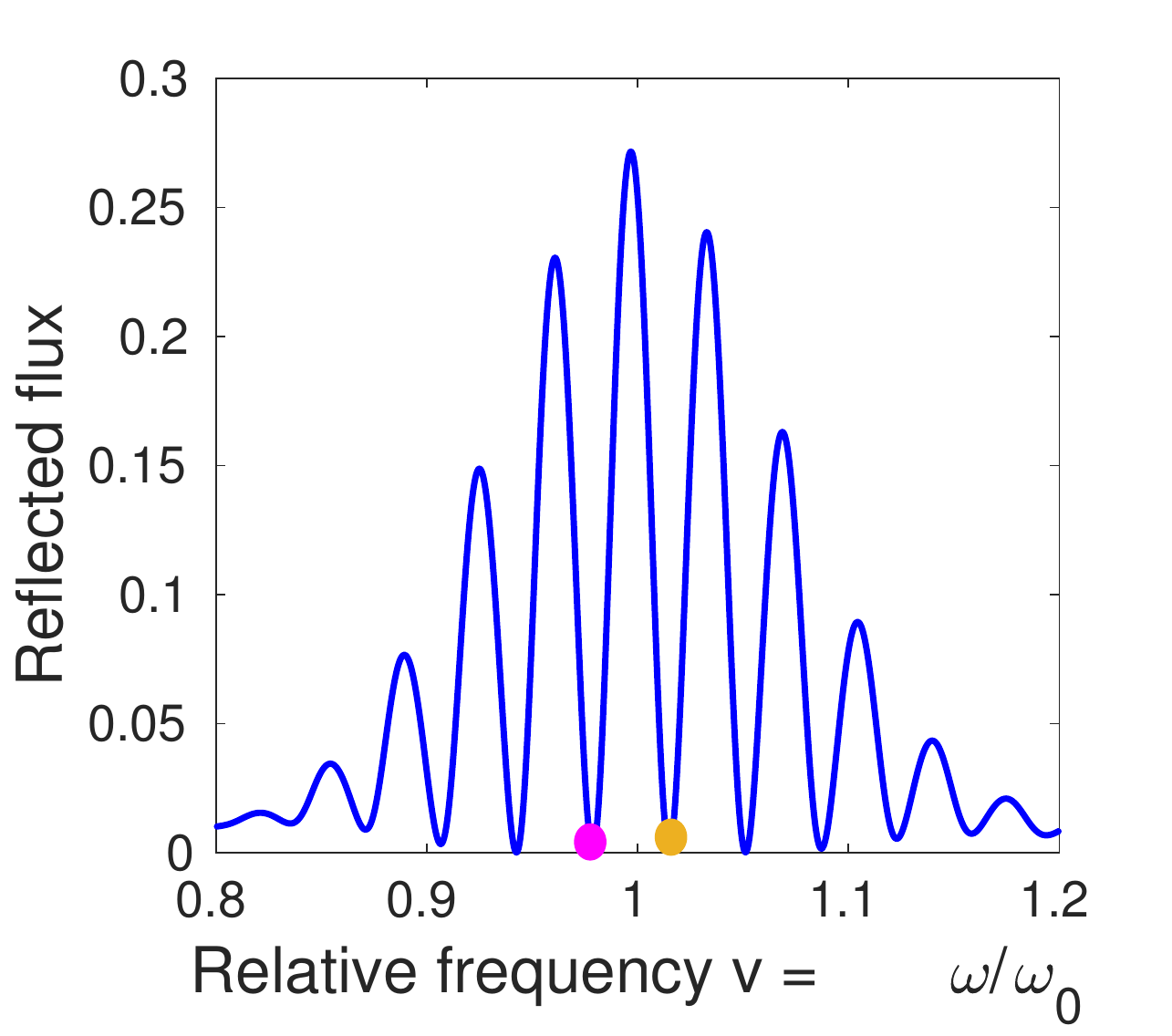}}
{\includegraphics[width=.49\columnwidth]{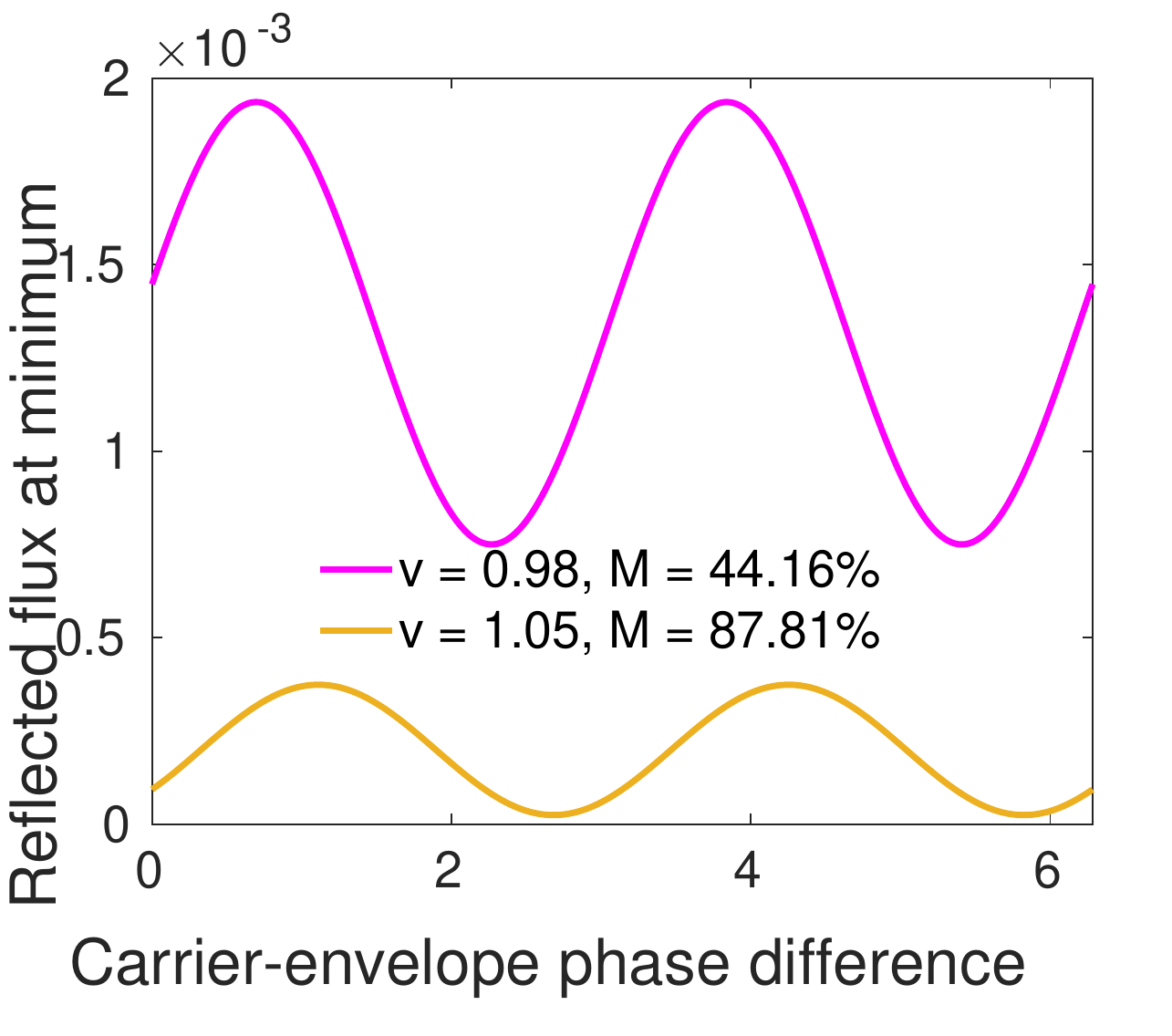}}
\caption{Dependence of the reflected flux on the CE phase difference at the minimum $v=0.98$ and $v=1.05,$ respectively. The refraction indices are $n_1=1$ $n_3=n_5=1.5,$ $\theta_1=\pi/8$ and 
the thickness of the two layers is $l_2=l_4=0.6$ nm.}\label{fig:CEP}
\end{figure}
\section{Concluding remarks}
In this paper we have derived, from first principles, the coupled system of equations describing the scattering of plane electromagnetic radiation fields on two parallel current sheets, which are embedded in three dielectric media. In this description the radiation field may represent ultrashort light pulses of arbitrary temporal shape and intensity, within the limit of the non-relativistic description of the local electron motions. This formalism yields a closed coupled set of two delay differential equations and a recurrence relation for the electronic velocities in the layers and for the reflected and transmitted field components, respectively. 
An exact analytic solution of this model is presented based on the Laplace transformation of the unknown time-dependent functions, without  any restriction on the physical parameters.  The eigenfrequencies of this dynamical system have been analyzed in detail. The main tool used in this analysis is the theory of singularly 
perturbed systems. Several numerical illustrations for the transmission and reflection properties of the two-layer system have been given and these have shown the temporal behavior of the outgoing fields in the case of few-cycle ultrashort incident pulses. The sensitivity of the resonant structure of the reflection (and transmission) coefficients against the carrier-envelope phase difference of the incoming pulse  provides a new way of measuring this  CE phase and  thus may be of immediate practical importance.


The analysis can be extended to the problem of several layers resulting in a larger system with more delays. 

This paper describes the propagation of p-polarized transverse magnetic (TM) waves but the analysis is analogous in the case of an s-polarized incoming transverse electric (TE) waves.

\section*{Acknowledgments}
The work of M. Polner was supported by the J\'anos Bolyai Research Scholarship of the Hungarian Academy of Sciences and by the UNKP-18-4 New National Excellence Program of the Ministry of Human Capacities. The ELI-ALPS project (GINOP-2.3.6-15-2015-00001) is supported by the European
Union and co-financed by the European Regional Development Fund. A. V\"or\"os-Kiss was supported by the UNKP-18-3 
New National Excellence Program of the Ministry of Human Capacities.
\appendix

\section{Dimensionless form of the equations}\label{s:dimensionless}
In this section, the equations \eqref{eq:282930} are made dimensionless. 
Denoting the velocities of the electrons by $v_j(t')=\dot\delta_{y_j}(t'),\ j=2,4$,  the second order system 
\eqref{eq:29}--\eqref{eq:30} is equivalent with a first order system. Introduce the dimensionless 
variables, denoted by a star
\begin{align}
t^*&=\frac{t'}{T},\quad v_i^*=\frac{v_i}{c}\quad
F^*=\frac{F}{F_0},\quad f_3^*=\frac{f_3}{F_0},\quad\Delta t^*_j=\frac{\Delta t_j}{T}=\frac{n_j h\cos \theta_j}{\lambda_0}, 
\end{align}
for $j=3,5$, where $T$ denotes the optical period and $F_0$ is the reference field strength. 
Inserting them into \eqref{eq:29}--\eqref{eq:30}, results in
\begin{subequations}\label{eq:vi_dimless}
\begin{align}
\frac{c}{T}\dot v^*_{2}(t^*) 
&= \frac{2c_1 c_3}{c_1+c_3}\left[
\frac{e}{m}F_0 F^*(t^*) - r_2 \omega_0 c v^*_{2}(t^*)
+\frac{e}{m}F_0 f_3^*(t^*)
\right],
\\[5pt]
\frac{c}{T}v^*_{4}(t^*) 
&=  \frac{2c_1 c_3}{c_1+c_3}
\left[
\frac{e}{m}F_0 F^*(t^*-\Delta t^*_3) - r_2 \omega_0c v^*_{2}(t^*-\Delta t^*_3)
\right]\nonumber
\\[5pt]
&+ \frac{c_1-c_3}{c_1+c_3} c_3 \frac{e}{m}
F_0f_3^*(t^*-\Delta t_3^*)
+ c_3 F_0\frac{e}{m} f_3^*(t^*+\Delta t_3^*).
\end{align}
\end{subequations}
Simplifying and substituting the dimensionless vector potential (intensity parameter) $a_0=(e F_0)/(mc \omega_0)$,  
$\Gamma_i=r_i\omega_0$,  with
\begin{equation}
r_i=\left(\frac{\omega_{p_i}}{\omega_0}\right)^2\frac{\pi l_i}{\lambda_0},\quad i=2,4
\end{equation}
and $\omega_0=2\pi/T$ into \eqref{eq:vi_dimless}, we obtain the dimensionless form of the system \eqref{eq:29}--\eqref{eq:30}
\begin{subequations}\label{eq:36_5}
\begin{align}
\dot v^*_{2}(t^*)
&= \frac{2c_1 c_3}{c_1+c_3}\left[
2\pi  a_0 F^*(t^*) - 2\pi r_2 v^*_{2}(t^*)
+ 2\pi a_0 f^*_{3}(t^*)
\right],
\\[5pt]
\dot v^*_{4}(t^*)
&= \frac{2c_1 c_3}{c_1+c_3}
\left[
2\pi a_0 F^*(t^*-\Delta t^*_3) - 2\pi r_2 v^*_{2}(t^*-\Delta t^*_3)
\right]\nonumber
\\[5pt]
&+ \frac{c_1-c_3}{c_1+c_3} c_3 2\pi a_0 f_3^*(t^*-\Delta t_3^*)
+ c_3 2\pi a_0 f_3^*(t^*+\Delta t_3^*).
\end{align}
\end{subequations}
Similarly, the recurrent equation can be nondimensionalized to obtain
\begin{align}\label{eq:f3_dimless}
f_3^*(t^*)
&= \frac{c_5-c_3}{c_5+c_3}
\cdot \frac{c_1-c_3}{c_1+c_3} f_3^*(t^*-2\Delta t_3^*)
\nonumber
\\[5pt]
&+ \frac{c_5-c_3}{c_5+c_3}
\cdot \frac{2c_1}{c_1+c_3}
\left[ F^*(t^*-2\Delta t_3^*)
- \frac{r_2}{a_0} v_2^*(t^*-2\Delta t_3^*)
\right]
\nonumber
\\[5pt]
&- \frac{2c_5}{c_5+c_3} \cdot \frac{r_4}{a_0} v_4^*(t^*-\Delta t_3^*).
\end{align}

We can conclude that the dimensional \eqref{eq:282930} and dimensionless \eqref{eq:36_5}--\eqref{eq:f3_dimless} systems have the same functional form. 
Therefore, the stars can be dropped from the non-dimensional equations. 
\allowdisplaybreaks
\section{Functions involved in $\Delta^{-1}(\epsilon,s)$}\label{app:parameters}
The functions $g(\epsilon,s)$ and the elements of the matrix function $\text{adj }\Delta(\epsilon,s)$ in \eqref{H_eps} are, respectively
\begin{align*}
g(\epsilon,s)=&\frac{\left(sa_2-a_1r_2(a_1-a_2)\right)(1-a_5)}{a_1-a_2}+(1+\epsilon s)(s+a_1r_2)e^{2\tau s}\\
&+(1+\epsilon a_1 r_2+\epsilon s)a_5r_4(a_1-a_2)e^{2\tau s}+a_2a_5r_4+a_1r_2r_4a_5(a_1-a_2)\frac{e^{2\tau s}-1}{s},
\end{align*}
and 
\begin{align*}
\text{adj }\Delta(\epsilon,s)_{11}=&\epsilon s^2+\left( 1+\frac{a_2(1-a_5)}{a_1-a_2}e^{-2\tau s}+(a_1-a_2)\epsilon a_5r_4\right)s\\
&+a_5r_4(a_1-a_2+a_2e^{-2\tau s} ),\\
\text{adj }\Delta(\epsilon,s)_{12}=&-a_1 a_5r_4 e^{-\tau s},\\
\text{adj }\Delta(\epsilon,s)_{21}=&-(1+\epsilon s)a_1 a_5r_2 e^{-\tau s},\\
\text{adj }\Delta(\epsilon,s)_{22}=&\epsilon s^2+\left( 1+\frac{a_2(1-a_5)}{a_1-a_2}e^{-2\tau s}+a_1r_2\epsilon \right)s\\
&+a_1r_2(1+(a_5-1)e^{-2\tau s} ),\\
\text{adj }\Delta(\epsilon,s)_{13}=& a_1a_0 s+a_1a_0a_5r_4(a_1-a_2),\\
\text{adj }\Delta(\epsilon,s)_{23}=&s a_0a_2 a_5 e^{-\tau s}-a_1a_0a_5r_2(a_1-a_2)e^{-\tau s},\\
\text{adj }\Delta(\epsilon,s)_{33}=& s^2+\left( a_1r_2+a_5r_4(a_1-a_2)\right) s+a_1r_2a_5r_4(a_1-a_2),\\
\text{adj }\Delta(\epsilon,s)_{31}=& \frac{a_1r_2e^{-2\tau s}}{a_0(a_1-a_2)}\left((1-a_5)s+a_5 r_4(a_1-a_2)\right),\\
\text{adj }\Delta(\epsilon,s)_{32}=& -\frac{a_5 r_4}{a_0} e^{-\tau s}\left(s+r_2 a_1\right).
\end{align*}
\section{Proof of Lemma \ref{l:spectrum}}\label{app:proof_lemma}
\begin{proof}
It is straightforward to calculate that $\det\Delta(\lambda)=0$ is equivalent to 
\begin{align}\label{detz0}
(\lambda +a_1 r_2)\left(\lambda +(a_1-a_2) r_4\right)=r_4 e^{-2\tau \lambda }\left(a_1 r_2(a_1-a_2)-a_2 \lambda \right).
\end{align}
Observe that $\lambda =0$ is a simple root of the characteristic equation \eqref{detz0}. 
Let $\lambda =x+iy$ in \eqref{detz0} and take the absolute value square of both sides to obtain
\begin{align}\label{roots} 
\left[\left(x+\frac{1}{2}\left(2\pi c_3r_4+a_1 r_2\right)\right)^2-\frac{1}{4}\left(2\pi c_3r_4+a_1 r_2\right)^2-y^2+2\pi c_3a_1r_2r_4\right]^2\nonumber\\
+4y^2\left(x+\frac{1}{2}\left(2\pi c_3r_4+a_1 r_2\right)\right)^2
=r_4^2e^{-4\tau x}\left(\left(2\pi c_3a_1r_2-a_2 x\right)^2+a_2^2y^2\right),
\end{align}
where $a_1-a_2=2\pi c_3$ is used. For any $y\in\mathbb{R}$, denote the left and the right hand sides of 
\eqref{roots} by $l(x)$ and $r(x)$, respectively. 
The aim is to show that for any $y\in\mathbb{R}$, all solutions of $l(x)=r(x)$ are non-positive. 

By straightforward calculations,
\begin{align*}
l(0)-r(0)&=\left(-y^2+2\pi c_3a_1r_2r_4\right)^2+y^2\left(2\pi c_3r_4+a_1 r_2\right)^2-
r_4^2\left(\left(2\pi c_3a_1r_2\right)^2+a_2^2y^2\right)\\[5pt]
&=y^4+y^2 r_4^2\left((2\pi c_3)^2-a_2^2\right)+y^2a_1^2r_2^2\\[5pt]
&=y^2\left(y^2+r_4^2\frac{4c_1c_3}{(c_1+c_3)^2}+a_1^2r_2^2\right)>0\quad \forall y\not=0.
\end{align*}
This shows also that $x=0$ is a solution of $l(x)=r(x)$ only if $y=0$. Hence, there are no purely imaginary roots of 
$\det\Delta(\lambda)=0$. 

The function $l(x)$ is a fourth order polynomial in $x$ with derivative
\[
l'(x)=4\left(x+\frac{1}{2}(2\pi c_3r_4+a_1r_2)\right)\left[ \left(x+\frac{1}{2}(2\pi c_3r_4+a_1r_2)\right)^2+y^2
-\frac{1}{4}\left(2\pi c_3 r_4-a_1r_2\right)^2\right],
\]
which has roots at
\[
x_0=-\frac{1}{2}\left(a_1r_2+2\pi c_3r_4\right),\ x_\pm=-\frac{1}{2}(a_1r_2+2\pi c_3r_4)\pm 
\frac{1}{2}\sqrt{(a_1r_2-2\pi c_3r_4)^2-4y^2}.
\]
If $y$ is such that $x_\pm$ are real, then $x_-\leq x_0\leq x_+<0$ and at $x_+$ the function $l$ has a 
local minimum. When $y$ is large enough, such that $x_\pm$ are complex, then $l'(x)=0$ has only one solution $x_0<0$, which is 
a global minimum point for $l$.  In both cases $l(x)$ is strictly increasing for $x>0$ and has a positive 
$l(0)$ intercept with the $y$-axis. 

To analyze the monotonicity of the right hand side function $r(x)$, take 
\[
r'(x)=-2r_4^2e^{-4\tau x}h(x),
\]
where
\[
h(x)=2\tau{{ a_2}}^{2}{x}^{2}
-\left(8\tau{ a_1}{ r_2}\pi { c_3}+{{ a_2}}\right)a_2x+2\tau{{ a_2}}^{2}{y}^{2}+2{ a_2}{ a_1}{ 
r_2}\pi { c_3}+8\tau a_1^2{{ r_2}}^{2}{\pi }^{2}c_3^2.
\]
The roots of $h$ are
\[
x_{1,2}={\frac {{ a_2}+8\tau{ a_1}{ r_2}\pi { c_3}\pm
\sqrt {{{ a_2}}^{2}-16{\tau}^{2}{{ a_2}}^{2}{y}^{2}}}{4\tau{
 a_2}}},
\]
which are real for those values of $y$ for which the discriminant is non-negative. If there are no real roots, then 
$h(x)>0$, hence $r'(x)<0$,  and therefore $r(x)$ is strictly decreasing. In this case, $r$ has a positive $r(0)$ 
intercept with the $y$-axis, and as $l(0)>r(0)$, it follows that $l(x)=r(x)$ can only be for $x<0$.  

Consider the case when $h$ has two distinct real roots $x_1<x_2$. If $x_1<x_2<0$, then $h(x)>0$ for $x>x_2$, hence  $r'(x)<0$ for $x>0$ and the same argument as before implies that $l(x)\not=r(x)$ for $x>0$. 

Assume that the largest root of $h$ is positive, i.e., $x_2>0$. Since $x_2$ is a local maximum of the right hand side function $r$, it is sufficient to show that $l(x_2)>r(x_2)$. To prove this assertion, we distinguish two 
cases. When $a_2>0$, then $0<x_1<x_2$, and when $a_2<0$, then $x_1<0<x_2$. In both cases, $l(x_2)-r(x_2)$ is considered as a function of $r_4$ and it is observed that this difference will have positive values for all $r_4>0$. This is a lengthy but straightforward calculation and thus is  omitted. 
\end{proof}

In Figure~\ref{char_roots_special_cases} and~\ref{f:char_roots} a few characteristic roots, located close to the imaginary axis, are plotted. When $n_1=n_3=n_5=1$,  $h=10\lambda_0$ and the angle of incidence is $\theta_1=\pi/3$, then the delay is $\tau=5$, see Figure~\ref{char_roots_special_cases} (left). 
The size of the delay can be increased in several ways. The indices of refraction can be changed to $1=n_1\not=n_3=n_5=1.5$, which lead to $\tau=12.25$, Figure~\ref{char_roots_special_cases} (right). An other way to increase the delay is by  increasing the distance $h$ between the layers, so that $\tau=12.25$, as before, see Figure~\ref{f:char_roots} (left).
We can observe that, as the delay increases, more eigenvalues get closer to the imaginary axis.

In Figure~\ref{f:char_roots} (right), some characteristic roots are plotted for the perturbed DDE system, with 
$\epsilon=10^{-6}$. The refractive indices are set as $n_1=1$, $n_3=1.1$ and $n_5=1.5$, $h=10\lambda_0$, which results in $\tau=6.78$.
\begin{figure}  
\centering 
    {\includegraphics[width=.49\columnwidth]{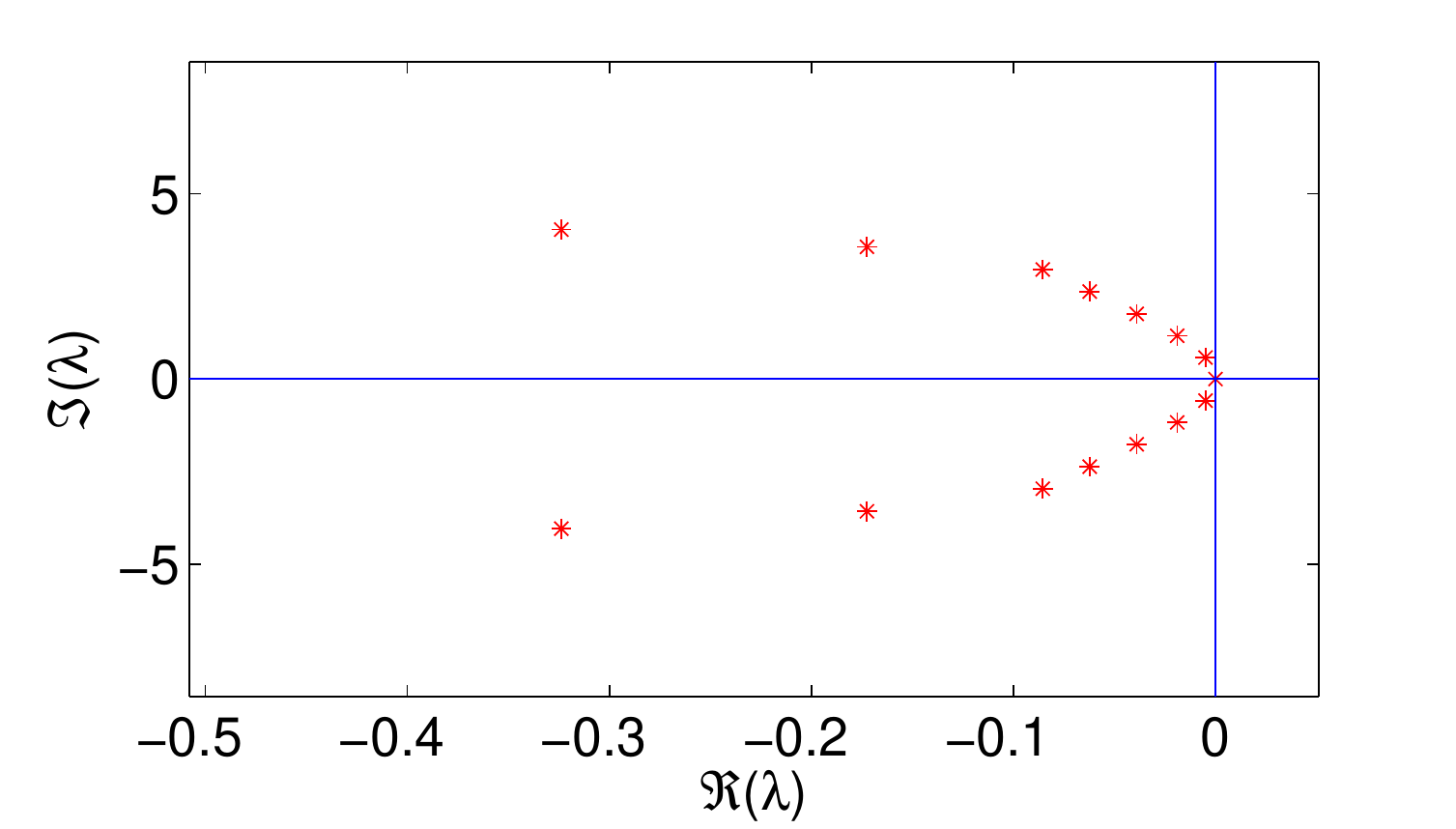}}  
   {\includegraphics[width=.49\columnwidth]{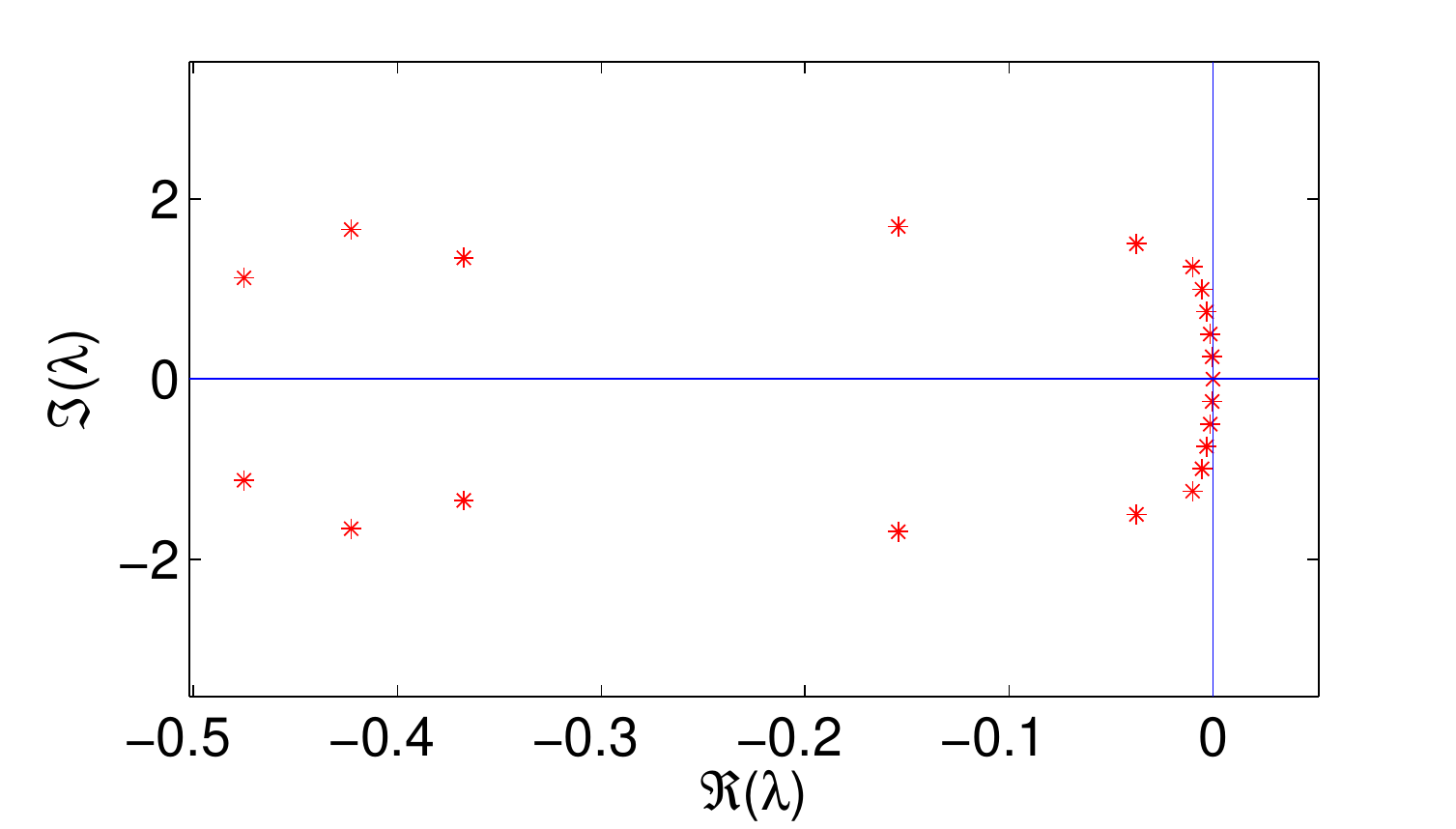} }
   {\caption{A few eigenvalues in the complex plane. When $n_1=n_3=n_5=1$,  then $\tau=5$ (left) and when $n_1=1$ and $n_3=n_5=1.5$,  
then $\tau=12.25$ (right).}\label{char_roots_special_cases}}
  \end{figure}
\begin{figure}  
\centering 
    {\includegraphics[width=.49\columnwidth]{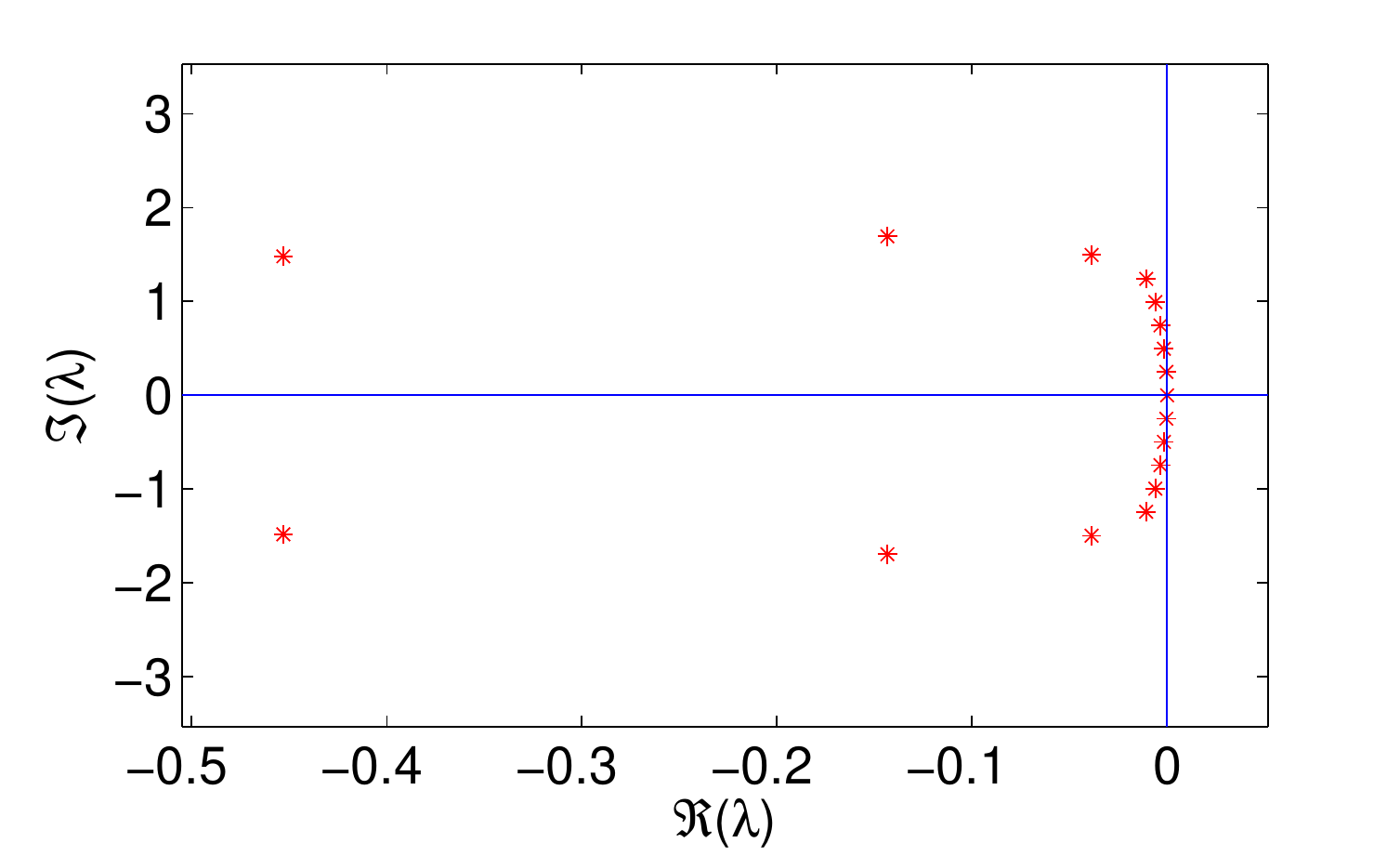}}  
   {\includegraphics[width=.49\columnwidth]{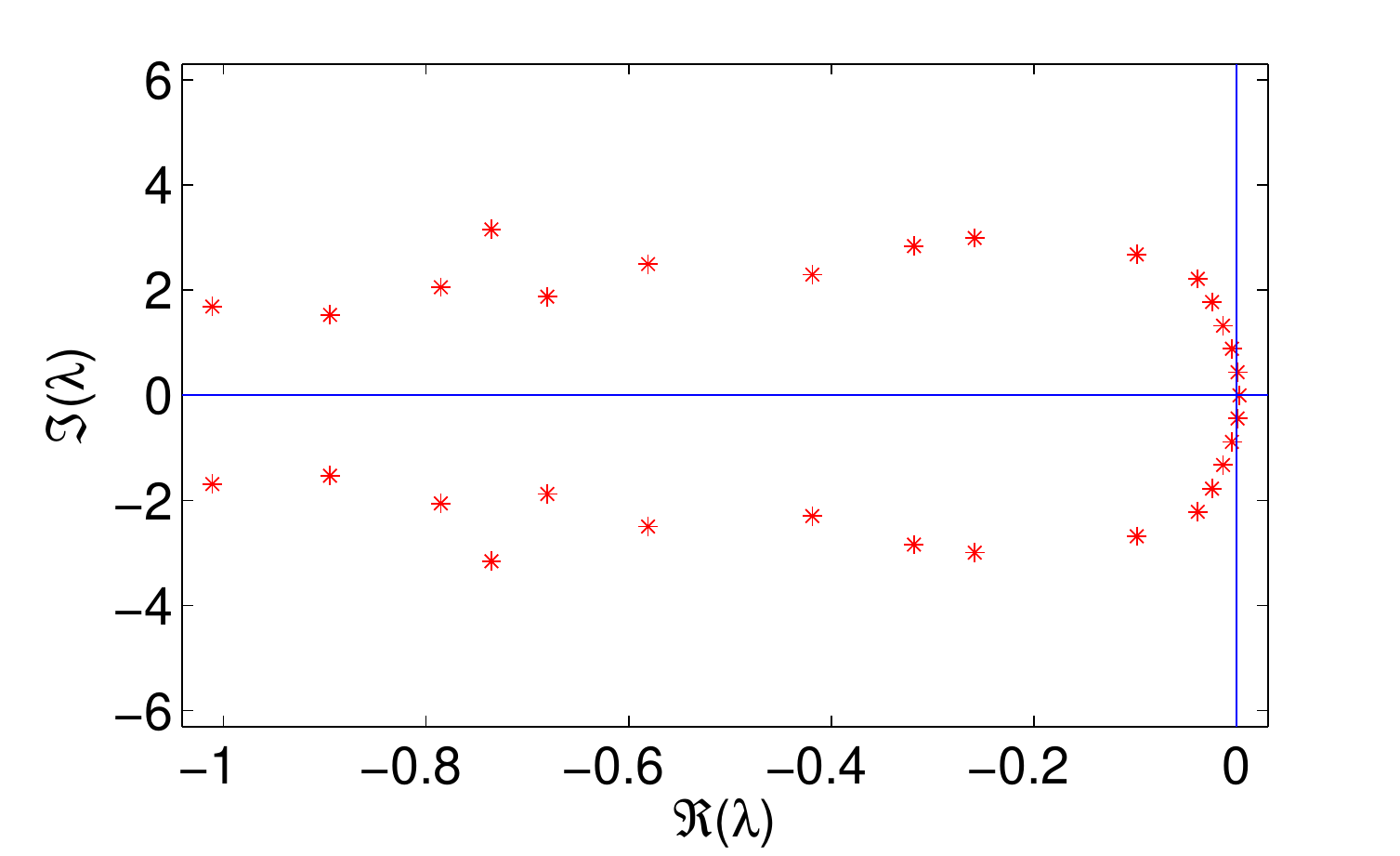} }
   {\caption{A few eigenvalues in the complex plane. When $n_1=n_3=n_5=1$ and $h=24.5 \lambda_0$,  then $\tau=12.25$ (left) and 
when $n_1=1$,  $n_3=1.1$,  $n_5=1.5$,  $h=10\lambda_0$,  then $\tau=6.78$. Here $\epsilon=10^{-6}$ (right).}\label{f:char_roots}}
  \end{figure}

\bibliographystyle{unsrt}

\bibliography{sample} 


\end{document}